\documentclass{LMCS}

\def\doi{9(3:3)2013}
\lmcsheading%
{\doi}
{1--47}
{}
{}
{Dec.~\phantom03, 2009}
{Aug.~13, 2013}
{}

\usepackage{amsmath}
\usepackage{amssymb}
\usepackage{tikz}
\usepackage{hyperref,enumerate}
\usepackage{paralist}

\usetikzlibrary{automata}
\newcommand{\mylabel}[1]{\label{#1}}

\newcommand{\NI}	{\mathbb N\cup\{\infty\}}
\newcommand{\osharp}	{{\omega\sharp}}
\newcommand{\val}	{\mathit{value}}
\newcommand{\unf}	{\text{unfold}}
\newcommand{\computation}	{\text{computation}}

\newcommand{\ignore}[1] {}

\newcommand{\ign}[1]{}

\newcommand{\set}[1]	{\left\{{#1}\right\}}
\newcommand{\nats}	{\mathbb{N}}

\newcommand{\alphabet}		{\mathbb{A}}
\newcommand{\alphabetB}		{\mathbb{B}}

\newcommand{\monoid}		{\mathbf{M}}

\newcommand{\sem}[1]	{[\![#1]\!]}

\newcommand{\semigroup}		{\mathbf{S}}

\newcommand{\gL}	{\mathcal{L}}
\newcommand{\gJ}	{\mathcal{J}}
\newcommand{\gR}	{\mathcal{R}}
\newcommand{\gH}	{\mathcal{H}}

\newtheorem{theorem}{Theorem}[section]
\newtheorem{definition}[theorem]{Definition}

\newtheorem{lemma}[theorem]{Lemma}
\newtheorem{proposition}{Proposition}
\newtheorem{corollary}[theorem]{Corollary}

\newtheorem{myfact}{Fact}[theorem] 
\newtheorem{remark}[myfact]{Remark}
\newtheorem{example}[myfact]{Example}

\newcommand{\ideal}		{{\downarrow}}

\newcommand{\coideal}		{{\uparrow}}

\newcommand{\FV}	{\mathrm{FV}}
\newcommand{\valuation}	{\mathrm{v}}
\newcommand{\structure}	{\mathcal{S}}

\newcommand{\intro}[1]	{\emph{#1}}

\begin{document}


\title[Regular cost functions, Part I]{Regular cost functions, Part I:\\logic and algebra over words}
\author[T.~Colcombet]{Thomas Colcombet}	
\address{Universit\'e Sorbonne Paris Cit\'e, {\sc Cnrs}, {\sc Liafa}}
\email{thomas.colcombet@liafa.univ-paris-diderot.fr}  
\thanks{Supported
	by the {\sc Anr} project {\sc Jade}: `Jeux et Automates, D\'ecidabilit\'e et Extensions'.
	The research leading to these results has received funding from the European Union's Seventh
	Framework Programme (FP7/2007-2013) under grant agreement n° 259454.
}

\ACMCCS{[{\bf Theory of computation}]:  Formal languages and automata theory---Automata extensions---Quantitative automata; Formal languages and automata theory---Regular languages; Formal languages and automata theory---Automata over infinite objects}

\begin{abstract}
The theory of regular cost functions is a quantitative
extension to the classical notion of regularity.
A cost function associates to each input a non-negative integer value (or infinity),
as opposed to languages which only associate to each input
the two values ``inside'' and ``outside''.
This theory is a continuation of the works on distance automata
and similar models. These models of automata 
have been successfully used for  solving the star-height problem,
the finite power property, the finite
substitution problem, the relative inclusion star-height problem and
the boundedness problem for monadic second-order logic over words.
Our notion of regularity can be -- as in the classical theory
of regular languages --
equivalently defined in terms of automata, expressions, algebraic recognisability,
and by a variant of the monadic second-order logic.
These equivalences are strict extensions of the corresponding
classical results.

The present paper introduces the cost monadic logic, the quantitative extension
to the notion of monadic second-order logic we use, and show that some~problems
of existence of bounds are decidable for this logic.
This is achieved by introducing the corresponding algebraic formalism:
stabilisation monoids.
\end{abstract}

\maketitle

\section{Introduction}

This paper introduces and studies a quantitative extension to the standard theory
of regular languages of words. It is the only quantitative extension (in which
quantitative means that the function described can take infinitely many values)
known to the author  in which the milestone equivalence for regular languages:
\begin{center}
accepted by automata = recognisable by monoids\\ = definable in monadic second-order logic =
definable by regular expressions
\end{center}
can be faithfully extended.

This theory is developed in several papers. The objective of the present one is
the introduction of the logical formalism, and its resolution using algebraic tools.
However, in this introduction, we try to give a broader panorama.

\subsection{Related works.}	
The theory of regular cost functions involves the use of automata (called $\mathtt B$- and $\mathtt S$-automata),
algebraic structures (called stabilisation monoids), a logic (called cost monadic logic),
and suitable regular expressions
(called $\mathtt B$- and $\mathtt S$-regular expressions). 
All these models happen to be of same expressiveness.
Though most of these concepts are new, some are very close to objects
known from the literature. As such, the present work is the continuation of several branches of research.

The general idea behind these works is that we want to represent functions, {\it i.e.}, quantitative
variants of languages, and that, ideally we want to keep strong decision results.
Works related to cost functions go in this direction,
where the quantitative notion is the ability to count, and the decidability results are concerned
with the existence/non-existence of bounds.

A prominent question in this theory is the star-height problem.
This story begins in 1963 when Eggan formulates the star-height decision problem \cite{eggan63}:
\begin{description}
\item[Input] A regular language of words~$L$ and a non-negative integer~$k$.
\item[Output] Yes, if there exists a regular expression\footnote{Regular
		expressions are built on top of letters using the language operation of
		concatenation, union, and Kleene star. This problem is sometime referred
		to as the restricted star-height problem, while the version also allowing
		complement is the generalised star-height problem, and has a very 
		different status.}
	using at most $k$ nesting of Kleene stars which defines $L$. No, otherwise.
\end{description}
Eggan proved that the hierarchy induced by~$k$ does not collapse,
but the decision problem itself was quickly considered as central in language theory, and as
the most difficult problem in the area.

Though some partial results were obtained by McNaughton, Dejean
and Sch\"utzenberger \cite{McNaughton67,DejeanS66},
it took twenty-five years before Hashiguchi came up
with a proof of decidability spread
over four papers \cite{Hashiguchi82a,Hashiguchi82b,Hashiguchi83,Hashiguchi88}.
This proof is notoriously difficult, and no clean exposition of it has ever
been presented.

Hashiguchi used in his proof the model of \emph{distance automata}.
A distance automaton is a finite state non-deterministic automaton running over words
which can count the number of occurrences of some ``special'' states.
Such an automaton associates to each word a natural number,
which is the least number of occurrences of special states among all the accepting runs
(or nothing if there is no accepting run over this input). The proof of Hashiguchi
relies on a very difficult reduction to the following \emph{limitedness} problem:
\begin{description}
\item[Input] A distance automaton.
\item[Output] Yes, if the automaton is \emph{limited}, {\it i.e.}, if the function it computes is
	bounded over its domain. No, otherwise.
\end{description}
Hashiguchi established the decidability of this problem \cite{Hashiguchi82b}.
The notion of distance automata and its relationship with the
tropical semiring (distance automata can be seen as automata over the tropical semiring,
{\it i.e.}, the semiring $(\NI,\min,+)$) has been the source of many investigations
\cite{Hashiguchi82a,Hashiguchi90,Hashiguchi00,Leung91,LeungP04,Simon78,Simon88,Simon94,Weber93,Weber94}.

Despite this research, the star-height problem itself remained not so well understood
for seventeen more years.
In 2005, Kirsten gave a much simpler and self-contained proof \cite{Kirsten05}.
The principle is to use a reduction to the limitedness problem for a
form of automata more general than distance automata,
called \emph{nested distance desert automata}.
To understand this extension, let us first look again at distance automata:
we can see a distance automaton as an automaton that has a counter
which is incremented each time a ``special'' state is encountered. The value attached
to a word by such an automaton is the minimum over all accepting runs
of the maximal value assumed by the counter.
Presented like this, a nested distance desert automaton is nothing but a distance automaton
in which multiple counters and reset of the counters are allowed
(with a certain constraint of nesting of counters).
Kirsten performed a reduction of the star-height problem to the limitedness of
nested distance desert automata which is much easier than the reduction of Hashiguchi.
He also proves that the limitedness problem of nested distance desert automata
is decidable. For this, he generalises the proof methods developed previously
by Hashiguchi, Simon and Leung for distance automata.
This work closes the story of the star-height problem itself.

The star-height problem is the king among the problems solved using this method.
But there are many other (difficult) questions that can be reduced to the limitedness
of distance automata and variants. Some of the solutions to these problems paved the
way to the solution of the star-height problem.

The \emph{finite power property} takes as input a regular language~$L$ and asks whether
there exists some positive integer~$n$ such that~$(L+\varepsilon)^n=L^*$.
It was raised by Brzozowski in 1966, and it took twelve years before being 
independently solved by Simon and Hashiguchi \cite{Simon78,Hashiguchi79}.
This problem is easily reduced to the limitedness problem for distance automata. 

The \emph{finite substitution problem} takes as input two regular languages~$L,K$, and
asks whether it is possible to find a finite substitution~$\sigma$ ({\it i.e.}, a morphism mapping each letter
of the alphabet of~$L$ to a finite language over the alphabet of~$K$) such that~$\sigma(L)=K$.
This problem was shown decidable independently by Bala and Kirsten by a reduction
to the limitedness of desert automata (a form of automata weaker than nested distance desert automata, but incomparable
to distance automata), and a proof of decidability of this latter problem \cite{Bala04,Kirsten04a}.

The \emph{relative inclusion star-height problem} is an extension
of the star height problem introduced and shown decidable
by Hashiguchi using his techniques \cite{Hashiguchi91}.
Still using nested distance desert automata, Kirsten gave
another, more elegant proof of this result \cite{Kirsten09}.

The \emph{boundedness problem} is a problem of model theory.
It consists of deciding if there exists a bound on the number of iterations
that are necessary for the fixpoint of a logical formula to be reached.
The existence of a bound means that the fixpoint can be eliminated by
unfolding its definition sufficiently many times. 
The boundedness problem is usually parameterised by the logic chosen
and by the class of models over which the formula is studied.
The boundedness problem for monadic second-order formulae over the class of finite words
was solved by a reduction to the limitedness problem of distance automata by Blumensath,
Otto and Weyer~\cite{BlumensathOttoWeyer09}.

One can also cite applications of distance automata in speech recognition~\cite{Mohri97,MohriPR02}, databases
\cite{GrahneThomo01}, and image compression \cite{CulikK93}. In the context of verification,
Abdulla, Krc\`al and Yi have introduced $\mathtt R$-automata, which correspond to nested distance desert
automata in which the nesting of counters is not required anymore \cite{AbdullaKY08}. They prove the
decidability of the limitedness problem for this model of automata.

Finally, L\"oding and the author have also pursued this branch of researches
in the direction of extended models.
In \cite{CSL08:colcombet-loeding-depth-mu-calculus}, the \emph{star-height problem over trees} has been solved,
by a reduction to the limitedness problem of nested distance desert automata over trees.
The latter problem was shown decidable in the more general case of alternating automata.
In \cite{ICALP08:colcombet-loeding-mostowski} a similar attempt has been tried for deciding the
Mostowski hierarchy of non-deterministic automata over infinite trees (the hierarchy induced by the
alternation of fixpoints). The authors show that it is possible to reduce this problem to
the limitedness problem for a form of automata that unifies nested distance desert automata
and parity tree automata. The latter problem is an important open question.

Boja\'nczyk and the author have introduced the notion of~$\mathtt B$-automata in \cite{LICS06:bojanczyk-colcombet},
a model which resembles much (and is prior to) $\mathtt R$-automata. The context was to show the decidability
of some fragments of the logic MSO+$\mathbb U$ over infinite words, in which MSO+$\mathbb U$ is the extension
of the monadic second order logic extended with the quantifier $\mathbb U X.\varphi$ meaning
``for all integers~$n$, there exists a set~$X$ of cardinality at least~$n$
such that~$\varphi$ holds''. 
From the decidability results in this work, it is possible to derive every other limitedness results over finite words.
However, the constructions are complicated and of non-elementary complexity.
Nevertheless, the new notion of $\mathtt S$-automata was introduced, a model dual to~$\mathtt B$-automata.
Recall that the semantics of distance automata
and their variants can be expressed as a minimum over all runs of the maximum of the value taken
by counters. The semantics of $\mathtt S$-automata is dual: it is defined as the maximum
over all runs of the minimum of the value taken by the counters at the moment of their reset.
Unfortunately, it is quite hard to compare in detail this work with all others.
Indeed, since it was oriented toward the study of a logic over infinite words,
the central automata are in fact~$\omega\mathtt B$ and $\omega\mathtt S$-automata: automata accepting languages of infinite
words that have an infinitary accepting condition constraining the asymptotic behaviour of the counters along the run.
This makes these automata very different\footnote{One must be careful: these automata are not related
	to $\mathtt B$ and $\mathtt S$-automata as, say, B\"uchi automata are related to automata over finite words.
	We warn the reader that these models cannot be thought as the extension of cost functions to infinite words.}.
Indeed, the automata in~\cite{LICS06:bojanczyk-colcombet}
accept languages while the automata in study here define functions.
For achieving this, the automata use an extra mechanism 
involving the asymptotic behaviors of counters for deciding whether an infinite
word should be accepted or not.
This extra mechanism has no equivalent in distance automata,
and is in some sense ``orthogonal'' to the machinery involved in cost functions.
For this reason, $\mathtt B$-automata and $\mathtt S$-automata in \cite{LICS06:bojanczyk-colcombet}
are just intermediate objects that do not have all the properties we would like.
In particular $\mathtt B$- and $\mathtt S$-automata in  \cite{LICS06:bojanczyk-colcombet}
are not equivalent. However, the principle of using two dual forms of automata
is an important concept in the theory of regular cost functions.
The study of MSO+$\mathbb U$ has been pursued in several directions.
Indeed, the general problem of the satisfaction of MSO+$\mathbb U$ is a challenging
open problem. One partial result concerns the decision of WMSO+$\mathbb U$
(the weak fragment in which only quantifiers over finite sets are allowed)
which is decidable \cite{Bojanczyk11}. However, the techniques involved in this
work are not directly related to cost functions.

The proof methods for showing the decidability of the limitedness problem
of distance automata and their variants, are also of much interest by themselves.
While the original proof of Hashiguchi is quite complex, a major advance has been achieved by
Leung who introduced the notion of stabilisation \cite{Leung87,Leung88}
	(see also \cite{Simon88} for an early overview).
The principle is to abstract the behaviour of the
distance automaton in a monoid, and further describe the semantics of the counter using an operator of stabilisation, {\it i.e.},
an operator which describes, given an element of the monoid, what would be the effect of iterating it
a ``lot of times''. This key idea was further used and refined by Simon, Leung, Kirsten, Abdulla,
Krc\`al and Yi. This idea was not present in \cite{LICS06:bojanczyk-colcombet},
and this is one explanation for the bad complexity of the constructions.

Another theory related to cost functions is the one developed by Szymon Toru\'nczyk in his
thesis \cite{PhDTorunczyk11}.
The author proposes a notion of recognisable languages of \emph{profinite words} which happen to be equivalent
to cost functions. Indeed, profinite words are infinite sequences of finite words (which are
convergent in a precise topology, the profinite topology). As such, a single profinite word
can be used as a witness that a function is not bounded. Following the principle of this correspondence,
one can see a cost function as a set of profinite words: the profinite words corresponding
to infinite sequences of words over which the function is bounded. 
This correspondence makes Toru\'ncyk's approach equi-expressive
with cost functions over finite words as far as decision questions are concerned.
Seen like this, this approach can be seen as the theory of cost functions
presented in a more abstract setting. Still, some differences have to be underlined.
On one side, the profinite approach, being more abstract, loses some precision.
For instance in the present work, we have a good understanding of the
precision of the constructions: namely each operation can be performed doing an
at most ``polynomial approximation\footnote{The notion of approximation may be misleading: the results are exact,
	but, since we are only interested in boundedness questions, we allow ourselves to perform
	some harmless distortions of the functions. This distortion is measured  by an  approximation parameter
	called the correction function.}''.
On the other side, the presentation in terms
of profinite languages eliminates the corresponding annoying details in the development
of cost functions: namely there is no more need to control the approximation at each step.
Another interesting point is that the profinite presentation points naturally to
extensions, which are orthogonal to cost functions, and are highly related to MSO+$\mathbb U$.
For the moment, the profinite approach has been developed for finite words only.
It is not clear for now how easy this abstract presentation can be used  for treating more complex
models, as it has been done for cost functions, {\it e.g.}, over finite trees \cite{LICS10:colcombet-loeding}.

\subsection{Survey of the theory.}
The theory of regular cost functions gives a unified and general
framework for explaining all objects, results and constructions
presented above (apart from the results in \cite{LICS06:bojanczyk-colcombet}
that are of a slightly different nature). It also allows to derive
new results.

Let us describe the contributions in more details.
\smallskip

\noindent
\emph{Cost functions.}
The standard notion of language is replaced by the new notion of cost function.
For this, we consider mappings from a set~$E$ to~$\NI$ (in practice~$E$
is the set of finite words over some finite alphabet)
and the equivalence relation~$\approx$ defined by~$f\approx g$ if:
\begin{center}
for all $X\subseteq E$, $f$ restricted to~$X$ is bounded iff $g$ restricted to~$X$ is bounded.
\end{center}
Hence two functions are equivalent if
it is not possible to distinguish them using arguments of existence of bounds.
A \intro{cost function} is an equivalence class for~$\approx$.
The notion of cost functions is what we use
as a quantitative extension to languages. Indeed, every language~$L$ can be identified
with (the equivalence class of) the function mapping words in~$L$ to the value~$0$, 
and words outside~$L$ to~$\infty$.
All the theory is presented in terms of cost functions. This means that all
equivalences are considered modulo the relation~$\approx$.
\medskip

\noindent
\emph{Cost automata.} A first way to define regular cost functions is to use cost automata, 
which come in two flavours, $\mathtt B$- and~$\mathtt S$-automata.
The $\mathtt B$-automata correspond in their simple form to~$\mathtt R$-automata \cite{AbdullaKY08} and in their simple
and hierarchical form to nested distance desert automata in~\cite{Kirsten06,Kirsten06habilitation}.
Those are also very close to $\mathtt B$-automata in \cite{LICS06:bojanczyk-colcombet}.
Following the ideas in \cite{LICS06:bojanczyk-colcombet}, we also use the dual variant of~$\mathtt S$-automata.
The two forms of automata, $\mathtt B$-automata and $\mathtt S$-automata, are equi-expressive
in all their variants, an equivalence that we call the \emph{duality theorem}.
Automata are not introduced in this paper.
\medskip

\noindent
\emph{Stabilisation monoids.}
The corresponding algebraic characterisation makes use of the new notion of stabilisation monoids.
A stabilisation monoid is a finite ordered monoid together with a stabilisation operation.
This stabilisation operation expresses what it means
to iterate ``a lot of times'' some element. The operator of stabilisation was introduced by
Leung \cite{Leung87,Leung88} and used also by Simon, Kirsten, Abdulla, Krc\`al and Yi
as a tool for analysing the behaviour of distance automata and their variants.
The novelty here lies in the fact that in our case, stabilisation is now part of the
definition of a stabilisation monoid.
We prove that it is possible to associate unique semantics to all stabilisation monoids.
These semantics are represented by means of computations.
A computation is an object describing how a word consisting of elements
of the stabilisation monoid can be evaluated into a value in the stabilisation
monoid.
This key result shows that the notion of stabilisation monoid has a ``meaning''
independent from the existence of cost automata (in the same way a monoid
can be used for recognising a language, independently
from the fact that it comes from a finite state automaton).
This notion of computations is easier to handle than the notion
of compatible mappings used in the conference version of this work
\cite{ICALP09:colcombet-cost-function}.
\medskip

\noindent
\emph{Recognisable cost functions.}
We use stabilisation monoids for defining the new notion of recognisable cost functions.
We show the closure of recognisable cost functions under min, max, and new operations
called inf-projection and sup-projection (which are counterparts to projection in the theory
of regular languages).
We also prove that the relation $\approx$ (in fact the correspoding preorder~$\preccurlyeq$)
is decidable over recognisable cost functions. This decidability result
subsumes many limitedness results from the literature.
This notion of recognisability for cost functions is equivalent to 
being accepted by the cost automata introduced above.
\medskip

\noindent
\emph{Extension of regular expressions.}
It is possible to define two forms of expressions, $\mathtt B$- and~$\mathtt S$-regular expressions,
and show that these are equivalent to cost automata.
These expressions were  already introduced in \cite{LICS06:bojanczyk-colcombet}
in which a similar result was established.
\medskip

\noindent
\emph{Cost monadic logic.}
The cost monadic (second-order) logic is a quantitative extension to monadic (second-order) logic.
It is for instance possible to define the diameter of a graph in cost monadic logic.
The cost functions over words definable in this logic
coincide with the regular cost functions presented above.
This equivalence is essentially the consequence of the closure properties of regular cost functions
(as in the case of regular languages), and no new ideas are required here.
The interest lies in the logic itself.
Of course, the decision procedure for recognisable cost function entails decidability results for
cost monadic logic.
In this paper, cost monadic logic is the starting point of our presentation,
and our central decidability result is Theorem~\ref{theorem:main} stating
the decidability of this logic.
%
\bigskip

\subsection{Content of this paper.}
This paper does not cover the whole theory of regular cost functions over words.
The line followed in this paper is to start from the logic ``cost monadic logic'',
and to introduce the necessary material for ``solving it over words''. This requires the
complete development of the algebraic formalism.


In Section~\ref{section:logic}, we introduce the new formalism of cost monadic logic,
and show what is required to solve it. In particular, we introduce the notion of
cost function, and advocate that it is useful to consider the logic under this view.
We state there our main decision result, Theorem~\ref{theorem:main}.

In Section~\ref{section:stabilisation-monoid} we present the underlying algebraic structure: stabilisation monoids.
We then introduce computations,
and establish the key results of existence (Theorem~\ref{theorem:exists-computation})
and uniqueness (Theorem~\ref{theorem:unicity-computations}) of the value computed
by computations.

In Section~\ref{section:recognisability}, we use stabilisation monoids for defining recognisable cost functions.
We show various closure results for recognisable cost functions as well as decision procedures.
Those results happen to fulfill the conditions required in Section~\ref{section:logic} for showing
the decidability of cost monadic logic over words.

In Section~\ref{section:conclusion} some arguments are given on the relationship with the models of automata, which
are not described in this document, and on how these different notions interact in the big picture.




\section{Logic}
\label{section:logic}

\subsection{Cost monadic logic}

Let us recall that monadic second-order logic (monadic logic for short)
is the extension of first-order logic with the ability to quantify over
sets ({\it i.e.}, monadic relations). 
Formally monadic formulae use \intro{first-order variables} ($x,y,\dots$),
and \intro{monadic variables} ($X,Y,\dots$), and it is allowed in
such formulae to quantify existentially and universally over both
first-order and monadic variables, to use every boolean connective, 
to use the membership predicate ($x\in X$),
and every predicate of the relational structure.
We expect from the reader basic knowledge concerning monadic logic.
\begin{example}\label{example:MSO}
The monadic formula~$\mathtt{reach}(x,y,X)$ over the signature containing the single binary predicate $\mathbf{edge}$ (signature of a digraph):
\begin{multline*}
\mathtt{reach}(x,y,X)::=\quad x=y\quad\vee\quad\forall Z\\
	\left(x\in Z\wedge\forall z,z'~(z\in Z\wedge z'\in X\wedge \mathbf{edge}(z,z')\rightarrow z'\in Z)\right)
		\quad\rightarrow\quad y\in Z
\end{multline*}
describes the existence of a path in a digraph from vertex~$x$ to vertex~$y$ such that
all edges appearing in the path end in~$X$.
Indeed, it expresses that either the path is empty, or every sets~$Z$ containing~$x$
and closed under taking edges ending in~$X$, also contains~$y$.
\end{example}

In \intro{cost monadic logic},
one uses a single extra variable~$N$ of a new kind, called the
\intro{bound variable}. It ranges over non-negative integers.
\intro {Cost monadic} logic is obtained from monadic logic by allowing the extra predicate $|X|\leq N$ -- in which $X$
is some monadic variable and $N$ the bound variable -- if and only if it appears \emph{positively} in the formula
({\it i.e.}, under the scope of an even number of negations).
The semantic of~$|X|\leq N$ is, as one may expect, to be
satisfied if (the valuation of) $X$ has cardinality at most (the valuation of) $N$. 
Given a formula~$\varphi$, we denote by~$\FV(\varphi)$ its free variables, the bound variable excluded.
A formula that has no free-variables--it may still use the bound variable--is called a \intro{sentence}.

We now have to provide a meaning to the formulae of cost monadic logic.
We assume some familiarity of the reader with logic terminology.
A \intro{signature} consists of a set of symbols~$R,S,\dots$. To each symbol is attached
a non-negative integer called its \intro{arity}.
A (relational) \intro{structure} (over the above signature)
$\structure=\langle U_\structure,R^\structure,\dots,R^\structure\rangle$ consists of a set~$U_\structure$
called the \intro{universe}, and for each symbol~$R$ of arity~$n$ of a
relation~$R^\structure\subseteq U_\structure^n$.
Given a set of variables~$F$, a valuation of~$F$ (over~$\structure$)
is a mapping~$\valuation$ which to each monadic variable~$X\in F$
associates a set~$\valuation(X)\subseteq U_\structure$, and to each
first-order variable~$x\in F$ associates an element~$\valuation(x)\in U_\structure$.
We denote by~$\valuation,X=E$ the valuation~$\valuation$ in which~$X$ is further mapped to~$E$.
Given a cost monadic formula~$\varphi$, a valuation~$\valuation$ of its free variable over a structure~$\structure$
and a non-negative integer~$n$,
we express by $\structure,\valuation,n\models\varphi$ the fact that the formula~$\varphi$
is satisfied over the structure~$\structure$ with valuation~$\valuation$ when the variable~$N$ takes the value~$n$.
Of course, if~$\varphi$ is simply a sentence, we just write~$\structure,n\models\varphi$.
We also omit the parameter~$n$ when~$\varphi$ is a monadic formula.

The positivity assumption required when using the predicate~$|X|\leq N$ has straightforward
consequences.
Namely, for all cost monadic sentences $\varphi$, all relational structures~$\structure$,
and all valuations~$\valuation$,
$\structure,\valuation,n\models\varphi$ implies $\structure,\valuation,m\models\varphi$ for all~$m\geq n$.

Instead of evaluating as true or false as done above,
we see a formula of cost monadic logic~$\varphi$ of free variables~$F$
as associating to each relational structure~$\structure$ and each valuation~$\valuation$ 
of the free variables a value in~$\NI$ defined by:
$$
\sem{\varphi}(\structure,\valuation)=\inf\{n~:~\structure,\valuation,n\models\varphi\}\ .
$$
This value can be either a non-negative integer, or~$\infty$ if no valuation of~$N$
makes the sentence true.
In case of a sentence~$\varphi$, we omit the valuation and simply write~$\sem\varphi(\structure)$.
Let us stress the link with standard monadic logic in the following fact:
\begin{myfact}
For all monadic formula~$\varphi$, and all relational structures~$\structure$,
\begin{align*}
\sem\varphi(\structure)&=\begin{cases}
	0&\text{if}~\structure\models\varphi\\
	\infty&\text{otherwise\ .}
	\end{cases}
\end{align*}
\end{myfact}

\begin{example}\label{example:MSOc}
The sentence $\forall X~|X|\leq N$ calculates the
size of a structure. More formally $\sem{\forall X~|X|\leq N}(\structure)$
equals $|U_\structure|$.

A more interesting example makes use of Example~\ref{example:MSO}.
Again over the signature of digraphs, the cost monadic sentence:
$$
\mathtt{diameter}::=\forall x,y~\exists X~|X|\leq N\wedge\mathtt{reach(x,y,X)}.
$$
defines the diameter of the di-graph: indeed, the diameter of a graph is the least~$n$
such that for all pairs of states~$x,y$, there exists a set of size at most~$n$
allowing to reach $y$ from~$x$ (recall that in the definition of~$\mathtt{reach}(x,y,X)$,
$x$ does not necessarily belong to~$X$, hence this is the diameter in the standard sense).
\end{example}

From now on, for avoiding some irrelevant considerations, we will consider the variant of
cost monadic logic in which a) only monadic variables are allowed,
b) the inclusion relation $X\subseteq Y$ is allowed, and c) each relation over elements is
raised to a relation over singleton sets.
Keeping in mind that each element can be identified with the unique singleton set containing
it, it is easy to translate cost monadic logic into this variant.
In this presentation, it is also natural to see the inclusion relation as any other relation.
We will also assume that the negations are pushed to the leaves of formulae as is usual.
Overall a formula can be of one of the following forms:
$$
R(X_1,\dots X_n)\quad|\quad\neg R(X_1,\dots X_n)\quad|\quad|X|\leq N\quad|\quad\varphi\wedge\psi\quad|\quad\varphi\vee\psi\quad|\quad 
	\exists X.\varphi\quad|\quad\forall X.\varphi
$$
in which~$\varphi$ and~$\psi$ are formulas, $R$ is some symbol of arity~$n$ which can possibly be~$\subseteq$
(of arity~$2$), and~$X,X_1,\dots,X_n$ are monadic variables.

So far, we have described the semantic of cost monadic logic from the standard notion of model.
There is another equivalent way to describe the meaning of formulae, by induction
on the structure.
The equations are disclosed in the following fact.
\begin{myfact}\label{fact:induction-logic}
Over a structure~$\structure$ and a valuation~$\valuation$, the following equalities hold:
\begin{align*}
\sem{R(X_1,\dots,X_n)}(\structure,\valuation)		&=\begin{cases}0 & \text{if}~R^\structure(\valuation(X_1),\dots,\valuation(X_n))\\
												\infty&\text{otherwise}\end{cases}\\
\sem{\neg R(X_1,\dots,X_n)}(\structure,\valuation)	&=\begin{cases}\infty & \text{if}~R^\structure(\valuation(X_1),\dots,\valuation(X_n))\\
												0&\text{otherwise}\end{cases}\\
\sem{|X|\leq N}(\structure,\valuation)				&=|\valuation(X)|\\
\sem{\varphi\vee\psi}(\structure,\valuation)			&=\min(\sem\varphi(\structure,\valuation),\sem\psi(\structure,\valuation))\\
\sem{\varphi\wedge\psi}(\structure,\valuation)		&=\max(\sem\varphi(\structure,\valuation),\sem\psi(\structure,\valuation))\\
\sem{\exists X~\varphi}(\structure,\valuation)			&=\inf\{\sem\varphi(\structure,\valuation,X=E)~:~E\subseteq U_\structure\}\\
\sem{\forall X~\varphi}(\structure,\valuation)			&=\sup\{\sem\varphi(\structure,\valuation,X=E)~:~E\subseteq U_\structure\}
\end{align*}
\end{myfact}

As it is the case for monadic logic, no property (if not trivial) is decidable for monadic logic in general.
Since cost monadic logic is an extension of monadic logic, one cannot expect anything to be
better in this framework.
However we are interested, as in the standard setting, to decide properties over a restricted class
$\mathcal{C}$ of structures. The class~$\mathcal C$ can typically be the class of
finite words, of finite trees, of infinite words (of length~$\omega$, or beyond) or of infinite trees.
The subject of this paper is to consider the case of finite words over a fixed finite alphabet.

We are interested in deciding properties concerning the function described by cost monadic formulae
over~$\mathcal C$. But what kind of properties?
It is quite easy to see that, given a cost monadic sentence~$\varphi$ and~$n\in\nats$,
one can effectively produce a monadic formula~$\varphi^n$ such that for all
structures~$\structure$,
$\structure\models\varphi^n$ iff $\sem\varphi(\structure)=n$ (such a translation
would be possible even without assuming the positivity requirement in the use of the
predicates $|X|\leq N$).
Hence, deciding questions of the form~``$\sem\varphi=n$'' can be reduced to the
standard theory.

Properties that cannot be reduced to the standard theory, and that we are
interested in, involve the existence of bounds. One says below that a function
$f$ is \intro{bounded over} some set~$X$ if there is some 
integer~$n$ such that~$f(x)\leq n$
for all~$x\in X$. We are interested in the following generic problems:
\begin{description}	
\item[Boundedness] Is the function~$\sem\varphi$ bounded over~$\mathcal{C}$?\\
	Or (variant), is~$\sem\varphi$ bounded over a regular subset of~$\mathcal{C}$?\\
	Or (limitedness),  is~$\sem\varphi$ bounded over~$\{\structure~:~\sem{\varphi}(\structure)\neq\infty\}$?
\item[Divergence] For all~$n$, do only finitely many~$\structure\in\mathcal{C}$ satisfy~$\sem\varphi(\structure)\leq n$?
	\\Said differently, are all sets over which~$\sem\varphi$ is bounded of finite
	cardinality?
\item[Domination] For all~$E\subseteq\mathcal{C}$, does $\sem\varphi$ bounded over~$E$ imply that $\sem\psi$
	is also bounded over~$E$?
\end{description}
All these questions cannot be reduced (at least simply) to questions in the standard theory.
Furthermore, all these questions become undecidable for very standard reasons as
soon as the requirement of positivity in the use of the new predicate~$|X|\leq N$ is removed.
In this paper, we introduce suitable material for proving their
decidability over the class~$\mathcal C$ of words.

 One easily sees that the domination question is in fact a joint extension of the boundedness question
(if one sets $\varphi$ to be always true, {\it i.e.}, to compute the constant function~$0$),
and the divergence question (if one sets~$\psi$ to be measuring the size of the structure,
{\it i.e.}, $\forall X~|X|\leq N$).
Let us remark finally that if~$\varphi$ is a formula of monadic logic, then the boundedness question
corresponds to deciding if~$\varphi$ is a tautology. If furthermore~$\psi$ is also monadic, then the domination
consists of deciding whether~$\varphi$ implies~$\psi$.

In the following section, we introduce the notion of cost functions, {\it i.e.}, equivalence classes
over functions allowing to omit discrepancies of the function described, while preserving
sufficient information for working with the above questions.

\subsection{Cost functions}
\label{subsection:cost-function}

In this section, we introduce the equivalence relation~$\approx$ over functions,
and the central notion of cost function.

A \intro{correction function} $\alpha$ is a non-decreasing mapping from~$\nats$ to~$\nats$
such that~$\alpha(n)\geq n$ for all~$n$.
From now on, the symbols $\alpha,$ $\alpha'\dots$ implicitly designate correction functions.
Given $x,y$ in~$\NI$, $x\preccurlyeq_\alpha y$
holds if~$x\leq \overline\alpha(y)$ in which~$\overline\alpha$ is the extension of~$\alpha$
with~$\overline\alpha(\infty)=\infty$. For every set~$E$, $\preccurlyeq_\alpha$
is extended to~$(\NI)^E$ in a natural way by~$f \preccurlyeq_\alpha g$
if $f(x)\preccurlyeq_\alpha g(x)$ for all~$x\in E$, or equivalently $f\leq\overline\alpha \circ g$.
Intuitively, $f$ is dominated by~$g$ after it has been ``stretched''
by~$\alpha$.
One also writes $f\approx_\alpha g$ if~$f\preccurlyeq_\alpha g$ and~$g\preccurlyeq_\alpha f$.
Finally, one writes $f \preccurlyeq g$ (resp. $f\approx g$) if~$f\preccurlyeq_\alpha g$
(resp. $f\approx_\alpha g$) for some~$\alpha$.
A \intro{cost function} (over a set~$E$) is an equivalence class of~$\approx$
({\it i.e.}, a set of mappings from~$E$ to~$\NI$).

Some elementary properties of~$\preccurlyeq_\alpha$ are:
\begin{myfact}\label{fact:preccurlyeq-elem}
If~$\alpha\leq\alpha'$ and~$f\preccurlyeq_\alpha g$, then~$f\preccurlyeq_{\alpha'} g$.
If~$f\preccurlyeq_\alpha g\preccurlyeq_{\alpha'} h$, then $f\preccurlyeq_{\alpha\circ\alpha'} h$.
\end{myfact}
The above fact allows to work with a single correction function at a time. 
Indeed, as soon as two correction functions~$\alpha$ and~$\alpha'$ are involved in the same proof,
we can consider the correction function~$\alpha''=\max(\alpha,\alpha')$. By the above fact,
it satisfies that $f\preccurlyeq_\alpha g$ implies~$f\preccurlyeq_{\alpha''} g$,
and~$f\preccurlyeq_{\alpha'} g$ implies~$f\preccurlyeq_{\alpha''} g$.

\begin{example}\label{example:approx}
Over~$\nats\times\nats$, maximum and sum are equivalent for the doubling correction function (for short, $(\max)\approx_{\times2}(+)$). 
Indeed, for all~$x,y\in\omega$,
$$\max(x,y)\leq x+y\leq 2\times\max(x,y)\ .$$

\noindent
Our next examples concern mappings from sequences of words to~$\nats$.
We have 
\begin{align*}
|~|_a&\preccurlyeq |~|\ ,&\text{and}\qquad|~|_a\not\preccurlyeq |~|_b\ ,
\end{align*}
where $a$ and $b$ are distinct letters, $|~|$ is the function mapping each word to its length
and $|~|_a$ the function mapping each word to the number of occurrences of the letter $a$ it contains.
Indeed we have $|~|_a\leq|~|$
but the set of words $a^*$ is a witness that $|~|_a\preccurlyeq_\alpha|~|_b$
cannot hold whatever is $\alpha$.

\noindent
Given words $u_1,\dots,u_k\in\{a,b\}^*$, we have
$$|u_1\dots u_k|_a\approx_\alpha \max(|K|,\max_{i=1\dots k}|u_i|_a)
\qquad\text{where}\quad K=\{i\in\{1,\dots,k\}~:~|u_i|_a\geq 1\}
$$
in which~$\alpha$ is the squaring function.
 Indeed, for one direction we just have to remark:
\begin{align*}
\max(|K|,\max_{i=1\dots k}|u_i|_a)&\leq |u_1\dots u_k|_a,
\end{align*}
and for the other direction we use:
\begin{align*}
|u_1\dots u_k|_a&\leq \sum_{i\in K}|u_i|_a\leq(\max(|K|,\max_{i=1\dots k}|u_i|_a))^2\ .
\end{align*}
\end{example}
The relation~$\preccurlyeq$ has other characterisations:
\begin{proposition}\label{proposition:preccurlyeq-characterisation}
For all $f,g$ from~$E$ to~$\NI$, the following items are equivalent:
\begin{enumerate}[(1)]
\item $f\preccurlyeq g$, \label{item:dominiation-preccurly}
\item $\forall n\in\nats\ \exists m\in\nats\ \forall x\in E~g(x)\leq n\rightarrow f(x)\leq m\ ,$ and;\label{item:domination-quantif}
\item for all $X\subseteq E$, $g|_X$ is bounded implies $f|_X$ is bounded.\label{item:domination-bounded}
\end{enumerate}
\end{proposition}
\begin{proof}
\emph{From (\ref{item:dominiation-preccurly}) to (\ref{item:domination-quantif}).}
Let us assume $f\preccurlyeq g$, {\it i.e.}, $f\preccurlyeq_\alpha g$ for some~$\alpha$.
Let~$n$ be some non-negative integer, and~$m=\alpha(n)$.
We have for all~$x$, that~$g(x)\leq n$
implies~$f(x)\leq(\alpha\circ g)(x)\leq \alpha(n)=m$, thus establishing the second statement.

\emph{From (\ref{item:domination-quantif}) to (\ref{item:domination-bounded}).}
Let~$X\subseteq E$ be such that~$g|_X$ is bounded. Let~$n$ be a bound of~$g$ over~$X$.
Item~(\ref{item:domination-quantif}) states the existence of~$m$ such
that~$\forall x\in E\ g(x)\leq n\rightarrow f(x)\leq m$.
In particular, for all~$x\in E$, we have~$g(x)\leq n$ by choice of~$n$, and hence $f(x)\leq m$.
Hence~$f|_X$ is bounded by~$m$.

\emph{From (\ref{item:domination-bounded}) to (\ref{item:dominiation-preccurly}).}
Let~$n\in\nats$, consider the set~$X_n=\{x~:~g(x)\leq n\}$.
The mapping $g$ is bounded over~$X_n$ (by~$n$), and hence by (\ref{item:domination-bounded}), $f$ is also bounded.
We set~$\alpha(n)=\max(n,\sup f(X_n))$. Since $X_n\subseteq X_{n+1}$, the  function~$\alpha$
is non-decreasing. Since furthermore~$\alpha(n)\geq n$, $\alpha$ is a correction function.
Let now~$x\in X$. If~$g(x)<\infty$, we have that~$x\in X_{g(x)}$ by definition of the $X$'s.
Hence~$f(x)\leq\sup f(X_{g(x)})=\alpha(g(x))$. 
Otherwise~$g(x)=\infty$, and we have~$f(x)\leq\overline\alpha(g(x))=\infty$.
Hence~$f\preccurlyeq_\alpha g$.
\end{proof}

The last characterisation shows that the relation~$\approx$
is an equivalence relation that preserves the existence of bounds.
Indeed, all this theory can be seen as a method for
proving the existence/non-existence of bounds.
One can also remark that the questions of boundedness, divergence,
and domination presented in the previous section,
are preserved under replacing the semantic of a formula
by an~$\approx$-equivalent function. Furthermore, the domination question
can be simply reformulated as~$\sem\varphi\succcurlyeq\sem\psi$.

We conclude this section by some remarks on the structure of the $\preccurlyeq$ relation.
Cost functions over some set~$E$  ordered by~$\preccurlyeq$  form a lattice.
Let us show how this lattice refines the lattice of subsets of~$E$ ordered by inclusion.
The following elementary fact shows that we can identify a subset of~$E$
with the cost function of its characteristic function
(given a subset~$X\subseteq E$, one denotes by~\intro{$\chi$}$_X$
its \intro{characteristic mapping} defined by~$\chi_X(x)=0$ if~$x\in X$, and $\infty$ otherwise):
\begin{myfact}
For all $X,Y\subseteq E$, $\chi_X\preccurlyeq\chi_Y$ iff $X\supseteq Y$.
\end{myfact}
In this respect, the lattice of cost functions is a refinement of the lattice of subsets
of~$E$ equipped with the superset ordering. Let us show that this refinement is strict.
Indeed, there is only one language~$L$ such that~$\chi_L$
does not have~$\infty$ in its range, namely~$L=E$, however, we will
show in Proposition~\ref{proposition:cost-functions-uncountable}
that, as soon as $E$ is infinite, there are uncountably many cost functions
which have this property of not using the value~$\infty$.
\begin{proposition}\label{proposition:cost-functions-uncountable}
If~$E$ is infinite, then there exist at least continuum many different cost functions from~$E$ to~$\nats$.
\end{proposition}
\begin{proof}
Without loss of generality, we can assume $E$ countable, and even, up
to bijection, that $E=\nats\setminus\{0\}$.
Let~$p_0,p_1,\dots$ be the sequence of all prime numbers. Every $n\in E$
is decomposed in a unique way as~$p_1^{n_1}p_2^{n_2}\dots$ in which all~$n_i$'s
are null but finitely many (with an obvious meaning of the infinite product).
For all~$I\subseteq\nats$, one defines the function~$f_I$ from~$\nats\setminus\{0\}$ to~$\nats$
for all~$n\in\nats\setminus\{0\}$ by:
\begin{align*}
f_I(n)&=\max\{n_i~:~i\in I,~n=p_1^{n_1}p_2^{n_2}\dots\}\ .
\end{align*}
Consider now two different sets~$I,J\subseteq\nats$.
This means---up to a possible exchange of the roles of~$I$ and~$J$---that there exists~$i\in I\setminus J$.
Consider now the set~$X=\{p_i^k~:~k\in\nats\}$.
Then, by construction,~$f_I(p_i^k)=k$ and hence~$f_I$ is not bounded over~$X$.
However, $f_J(p_i^k)=0$ and hence~$f_J$ is bounded over~$X$.
It follows by Proposition~\ref{proposition:preccurlyeq-characterisation} that~$f_I$ and~$f_J$
are not equivalent for~$\approx$. We can finally conclude that---since there exist continuum many subsets of~$\nats$---
there is at least continuum many cost functions over~$E$ which do not use value $\infty$.
\end{proof}



\subsection{Solving cost monadic logic over words using cost functions}

As usual, we see a word as a structure, the universe of
which is the set of positions in the word (numbered from~$1$),
equipped with the ordering relation~$\leq$,
and with a unary relation for each letter of the alphabet that we interpret as the set of
positions at which the letter occur.
Given a set of monadic variables~$F$, and a valuation~$\valuation$ of~$F$
over a word~$u=a_1\dots a_k\in\alphabet^*$, we denote by~$\langle u,\valuation\rangle$
the word~$c_1\dots c_k$ over the alphabet~$\alphabet_F=\alphabet\times\{0,1\}^F$ such
that for all position~$i=1\dots k$, $c_i=(a_i,\delta_i)$
in which~$\delta_i$ maps~$X\in F$ to~$1$ if~$i\in\valuation(X)$,
and to $0$ otherwise.

It is classical that given a monadic formula~$\varphi$ with free variables~$F$, the language
$$
L_\varphi=\{\langle u,\valuation\rangle~:~u,\valuation\models\varphi\}\subseteq\alphabet_F^*
$$
is regular. The proof is done by induction on the formula.
It amounts to remark that to the constructions of the logic, namely
disjunction, conjunction, negation and existential quantification,
correspond naturally some language theoretic operations, namely union, intersection,
complementation and projection.
The base cases are obtained by remarking that the relations of ordering,
inclusion, and letter, also correspond to regular languages.

We use a similar approach. To each cost monadic formula~$\varphi$ with
free variables~$F$ over the signature of words over~$\alphabet$,
we associate the cost function~$f_\varphi$ over~$\alphabet_F$ defined by
$$
f_\varphi(\langle u,\valuation\rangle)=\sem\varphi(u,\valuation)\ .
$$
We aim at solving cost monadic logic by providing an explicit representation to the
cost functions $f_\varphi$.
For reaching this goal, we need to define a family
of cost functions~$\mathcal F$
that contains suitable constants, has effective closure properties
and decision procedures.

The first assumption we make is the closure under \intro{composition with a morphism}.
{\it I.e.}, let~$f$ be a cost function in~$\mathcal F$ over~$\alphabet^*$
and~$h$ be a morphism from~$\alphabetB^*$ ($\alphabetB$ being another alphabet)
to~$\alphabet^*$, we require~$f\circ h$ to also belong to~$\mathcal F$.
In particular, this operation allows us to change the alphabet, and hence to add
new variables when required. It corresponds to the closure under
inverse morphism for regular languages.

Fact~\ref{fact:induction-logic} gives us a very precise idea of the constants we need.
The constants correspond to the formulae of the form~$R(X_1,\dots,X_n)$ as well as their
negation. As mentioned above, for such a formula~$\varphi$, $L_\varphi$ is regular.
Hence, it is sufficient for us to require that the characteristic
function~$\chi_L$ belongs to~$\mathcal F$ for each regular language~$L$.
The remaining constants correspond to the formula~$|X|\leq N$.
We have that~$f_{|X|\leq N}(\langle u,X=E\rangle)=|E|$.
This corresponds to counting the number of occurrences of letters from~$\alphabet\times\{1\}$
in a word over~$\alphabet\times\{0,1\}$. Up to a change of alphabet (thanks to the closure under composition with a morphism)
it will be sufficient for us that~$\mathcal F$ contains
the function~``$\mathrm{size}$'' which maps each word~$u\in\{a,b\}^*$ to $|u|_a$.

Fact~\ref{fact:induction-logic} also gives us a very precise idea of the closure properties we need.
We need the closure under $\min$ and~$\max$ for disjunctions and conjunctions.
For dealing with existential and universal quantification, we need the new operations
of \intro{$\inf$-projection} and \intro{$\sup$-projection}.
Given a mapping $f$ from~$\alphabet^*$ to~$\NI$ and a mapping~$h$ from~$\alphabet$
to~$\mathbb{B}$ that we extend into a morphism from~$\alphabet^*$ to~$\mathbb{B}^*$
($\mathbb{B}$ being another alphabet)
the \intro{inf-projection} of~$f$ with respect to~$h$ is
the mapping $f_{\inf,h}$ from~$\mathbb{B}^*$ to~$\NI$ defined for all~$v\in\mathbb{B}^*$ by:
\begin{align*}
f_{\inf,h}(v)&=\inf\{f(u)~:~h(u)=v\}\ .
\end{align*}
Similarly, the \intro{sup-projection} of~$f$ with respect to~$h$ is
the mapping $f_{\sup,h}$ from~$\mathbb{B}^*$ to~$\NI$ defined for all~$v\in\mathbb{B}^*$ by:
\begin{align*}
f_{\sup,h}(v)&=\sup\{f(u)~:~h(u)=v\}\ .
\end{align*}

We summarise all the requirements in the following fact.
\begin{myfact}\label{fact:schema-cost-monadic}
Let~$\mathcal F$ be a class of cost functions  over words such that:
\begin{enumerate}[(1)]
\item for all regular languages $L$, $\chi_L$ belongs to~$\mathcal F$, \label{item:F-reg}
\item $\mathcal F$ contains the cost function ``$\mathrm{size}$'',\label{item:F-size}
\item $\mathcal F$ is effectively closed under composition with a morphism, $\min$, $\max$, $\inf$-projection and $\sup$-projection,
	\label{item:F-closure}
\item $\preccurlyeq$ is decidable over~$\mathcal F$,
	\label{item:F-decidability}
\end{enumerate}
then the boundedness, divergence and domination problems are decidable for cost monadic logic over words.
\end{myfact}


The remainder of the paper is devoted to the introduction of the class of recognisable cost functions,
and showing that this class satisfies all the assumptions of Fact~\ref{fact:schema-cost-monadic}.
In particular, Item~\ref{item:F-reg} is established as Example~\ref{example:language-recognisable}.
Item~\ref{item:F-size} is achieved in Example~\ref{example:recognisable-size}.
Item~\ref{item:F-closure} is the subject of Fact~\ref{fact:recognisable-composition-morphism},
Corollary~\ref{corollary:rec-min-max} and
Theorems~\ref{theorem:inf-projection} and \ref{theorem:sup-projection}.
Finally, Item~\ref{item:F-decidability} is established in Theorem~\ref{theorem:rec-domination}.

Thus we deduce our main result.
\begin{theorem}\label{theorem:main}
The domination relation is decidable for cost-monadic logic over finite words.
\end{theorem}

All these results are established in Section~\ref{section:recognisability}.
However, we need first to introduce the notion of stabilisation monoids, as well as some of
its key properties. This is the subject of Section~\ref{section:stabilisation-monoid}.


\section{The algebraic model: stabilisation monoids}
\label{section:stabilisation-monoid}

The purpose of this section is to describe the algebraic model of stabilisation monoids.
This model has, a priori, no relation with the previous
section. However, in Section~\ref{section:recognisability}, in which we define the
notion of a recognisable cost function, we will use this model of describing cost functions.

The key idea---an idea directly inspired from the work of Leung, Simon and Kirsten---is
to develop an algebraic notion (the stabilisation monoid) in which
a special operator (called the stabilisation, $\sharp$) allows to express what happens
when we iterate ``a lot of times'' some element. In particular, it says whether we
should count or not the number of iterations of this element. The terminology
``a lot of times'' is very vague, and for this reason such a formalism cannot describe
precisely functions. However, it is perfectly suitable for describing cost functions.

The remaining part of the section is organised as follows.
We first introduce the notion of stabilisation monoids in Section~\ref{subsection:stabilisation-monoid},
paying a special attention to give it an intuitive meaning.
In Section~\ref{subsection:computations}, we introduce
the key notions of computations, under-computations and over-computations,
as well as the two central results of existence of computations (Theorem~\ref{theorem:exists-computation})
and ``unicity'' of their values (Theorem~\ref{theorem:unicity-computations}).
These notions and results form the main technical core of this work.
Then Section \ref{subsection:proof-exists-computation} is devoted to the proof of Theorem~\ref{theorem:exists-computation},
and Section~\ref{subsection:proof-unicity-computations} to the proof of Theorem~\ref{theorem:unicity-computations}.

\subsection{Stabilisation monoids}
\label{subsection:stabilisation-monoid}

A \intro{semigroup}~$\semigroup=\langle S,\cdot\rangle$  is a set~$S$ equipped
with an associative operation~`$\cdot$'.
A \intro{monoid} is a semigroup such that the
product has a \intro{neutral element~$1$}, {\it i.e.}, such that~$1\cdot x=x\cdot 1=x$
for all~$x\in S$.
Given a semigroup~$\semigroup=\langle S,\cdot\rangle$,
we extend the product to products of arbitrary length by defining~\intro{$\pi$} from~$S^+$
to~$S$ by~$\pi(a)=a$ and~$\pi(ua)=\pi(u)\cdot a$. If the semigroup is a monoid of neutral element~$1$,
we further set~$\pi(\varepsilon)=1$.
All semigroups are monoids, and conversely it is sometimes convenient to transform a semigroup $\semigroup$ into
a monoid $\semigroup^1$ simply by the adjunction of a new neutral element $1$.

An idempotent in~$\semigroup$ is an element~$e\in S$ such that~$e\cdot e=e$.
We denote by $E(\semigroup)$ the set of idempotents in~$\semigroup$.
An \intro{ordered semigroup} $\langle S,\cdot,\leq\rangle$ is a semigroup $\langle S,\cdot\rangle$
together with an order~$\leq$ over~$S$ such that the product~$\cdot$ is \intro{compatible} with~$\leq$;
{\it i.e.}, $a\leq a'$ and~$b\leq b'$ implies ~$a\cdot b\leq a'\cdot b'$.
An \intro{ordered monoid} is an ordered semigroup, the underlying semigroup of which is a monoid.

We are now ready to introduce the new notions of stabilisation semigroups and stabilisation monoids.
\begin{definition}
A  \intro{stabilisation semigroup} $\langle S,\cdot,\leq,\sharp\rangle$
is a \emph{finite} ordered semigroup $\langle S,\cdot,\leq\rangle$
together with an operator~\intro{$\sharp$}$:E(\semigroup)\rightarrow E(\semigroup)$ (called the \intro{stabilisation}) such that:
\begin{iteMize}{$\bullet$}
\item for all~$e\leq f$ in~$E(\semigroup)$, $e^\sharp\leq f^\sharp$;
\item for all~$a,b\in S$ with~$a\cdot b\in E(\semigroup)$ and $b\cdot a\in E(\semigroup)$,
	 $(a\cdot b)^\sharp=a\cdot(b\cdot a)^\sharp\cdot b$;\footnote{This
	equation states that~$\sharp$ is a \intro{consistent mapping} in the sense of \cite{Kirsten05,Kirsten06habilitation}.}
\item for all~$e\in E(\semigroup)$, $e^\sharp\leq e$;
\item for all~$e\in E(\semigroup)$, $(e^\sharp)^\sharp=e^\sharp$.
\end{iteMize}
It is called a \intro{stabilisation monoid} if furthermore $\langle S,\cdot\rangle$ is a monoid and $1^\sharp=1$ 
in which $1$ is the neutral element of the monoid.
\end{definition}
The intuition is that $e^\sharp$ represents what is the value of~$e^n$
when~$n$ becomes ``very large''. Some consequences of the definitions, namely%
	\footnote{Indeed, $e^\sharp = (e\cdot 1)^\sharp = e\cdot (1\cdot e)^\sharp\cdot 1=e\cdot e^\sharp$
	using consistency. In the same way $e^\sharp = e^\sharp\cdot e$.
	Since $\sharp$ maps idempotents to idempotents,
	$e^\sharp=e^\sharp\cdot e^\sharp$ is obvious, and $(e^\sharp)^\sharp$ is also by definition.}
$$
\text{for all}~e\in E(\semigroup),\qquad
e^\sharp=e\cdot e^\sharp=e^\sharp\cdot e=e^\sharp\cdot e^\sharp=(e^\sharp)^\sharp\ ,
$$
make perfect sense in this respect: repeating ``a lot of~$e$'s'' is equivalent to seeing one~$e$ followed
by ``a lot of~$e$'s'', etc\dots This meaning of $e^\sharp$ is in some sense a limit behaviour.
This is an intuitive reason why $\sharp$ is not used for non-idempotent elements. Consider
for instance the element $1$ in $\mathbb Z/2\mathbb  Z$. Then iterating it yields $0$ at even iterations,
and $1$ at odd ones. This alternation prevents to giving a clear meaning to what is the result
of ``iterating a lot of times'' $1$.

However, this view is incompatible with the classical view on monoids, in which by induction,
if~$e\cdot e=e$, then~$e^n=e$ for all~$n\geq 1$.
The idea in stabilisation monoids is that the product is something that
cannot be iterated ``a lot of times''. For this reason, considering that for all~$n\geq 1$, $e^n=e$
is correct for ``small values of $n$'', but becomes ``incorrect'' for ``large values of $n$''.
The value of~$e^n$ is $e$ if~$n$ is ``small'', and it is~$e^\sharp$
if~$n$ is ``big''.
Most of the remainder of the section is devoted to the formalisation of this intuition, via the use of 
the notion of computations.

Even if the material necessary for  working with stabilisation monoids
has not been yet provided, it is already possible to give some examples of stabilisation
monoids that are constructed from an informal idea of their intended meaning.
\begin{example}\label{example:stabilisation-monoid-counta}
In this example we start from an informal idea
of what we would like to compute, and construct a stabilisation monoid from it.
The explanations have to remain informal at this point in the exposition of the theory.
However, this example should illustrate how we can already reason easily at this level
of understanding. 

Imagine you want, among words over~$a$ and~$b$, to separate the ones that
possess ``a lot of occurrences of~$a$'s'' from the ones that have only
``a few occurrences of~$a$'s'', {\it i.e.}, imagine you want to
describe a stabilisation monoid that ``counts'' the number of occurrences of~$a$'s.

For doing this, we should separate three ``kinds'' of words:
\begin{iteMize}{$\bullet$}
\item the kind of words with no occurrence of~$a$; let~$b$ be the corresponding element in the stabilisation monoid (since
the word~$b$ is of this kind),
\item the kind of words with at least one occurrence of~$a$, but only ``a few'' such occurrences;
		let $a$ be the corresponding element in the stabilisation monoid (since the word~$a$ is of this kind),
\item the kind of words with ``a lot of occurrences of~$a$'s''; let~$0$ be the corresponding element in the stabilisation monoid.
\end{iteMize}
The words that we intend to separate---the ones with a lot of~$a$'s---are the ones of kind~$0$.
With these three elements known, let us complete the definition of the stabilisation monoid.

Of course, iterating twice, or ``many times'', words which contain no occurrences of~$a$ yields
words with no occurrences of the letter~$a$. We capture this with the equalities~$b=b\cdot b=b^\sharp$.

Now words that have at least one~$a$, but only a few number of occurrences of~$a$, should not
be affected by appending~$b$ letters to their left or to their right. {\it I.e.}, we set~$a=b\cdot a=a\cdot b$.
Even more, appending a word with ``few $a$'s'' to another word with ``few $a$'s'' does also give
a word with ``few $a$'s''. {\it I.e.}, we set~$a\cdot a=a$.

However, if we iterate ``a lot of times'' a word with at least one occurrence of~$a$, it yields a word
with ``a lot of $a$'s''. Hence we set~$a^\sharp=0$.
We also easily get the equations $b\cdot 0=0\cdot b= a\cdot 0=0\cdot a=0\cdot 0=0^\sharp=0$
by inspecting all situations. We reach the following description of the stabilisation monoid:
\begin{center}
$\begin{array}[b]{l|cccc|c}
\cdot&~b~&~a~&~0~&~~~&~\sharp~\\
\hline
b~	&b&a&0&&b\\
a	&a&a&0&&0\\
0	&0&0&0&&0
\end{array}$
\end{center}
The left part describes the product operation~`$\cdot$', while the rightmost column gives the value of stabilisation~$\sharp$
(in general, this column may be partially defined since stabilisation is defined only for idempotents).

To complete the definition, we need to define the ordering over~$\{b,a,0\}$.
The least we can do is setting~$0\leq a$ and~$x\leq x$ for all~$x\in\{b,a,0\}$.
This is mandatory, since by definition of a stabilisation monoid~$a^\sharp\leq a$,
and we have $a^\sharp=0$. The reader can check that all the properties that we
expect from a stabilisation monoid are now satisfied.

The intuition behind the ordering is that depending on what we mean by ``a lot of $a$'s'', the same word
can be of kind~$a$ or of kind~$0$. For instance the word~$a^{100}$ is of kind~$a$ if we consider that~$100$ 
is ``a few'', while it is of kind~$0$ if we consider that~$100$ is ``a lot''.
For this reason, there is a form of continuum that allows to
go from~$a$ to~$0$. The order~$\leq$ captures this relationship between elements. 

It is sometimes convenient to present a stabilisation monoid by a form of Cayley graph:
\begin{center}
\begin{tikzpicture}[node distance=1.4cm,auto]
  \tikzstyle{every state}=[initial text=,very thick,draw=red!50,minimum size=2mm]
  \node[state] (b)					{$b$};
  \node[state] (a)	[right of=b]		{$a$};
  \node[state] (o)	[right of=a]		{$0$};
  \path[->] (b)	edge [in=60,out=120,loop]	node {$b$} ();
  \path[->] (b)	edge [in=180,out=240,loop,double] ();
  \path[->] (b)	edge			node {$a$} (a);
  \path[->] (b)	edge [bend angle=60,bend right] node {$0$} (o);
  \path[->] (a)	edge [in=60,out=120,loop]	node {$a,b$} ();
  \path[->] (a)	edge [bend right]			node {$0$} (o);
  \path[->] (o)	edge [in=50,out=110,loop]	node {$0,a,b$} ();
  \path[->] (a)	edge [bend left,bend angle=30,double] (o);
  \path[->] (o)	edge [in=-60,out=0,loop,double] ();
\end{tikzpicture}
\end{center}
As in a standard Cayley graph, there is an edge labeled by~$y$ going
from every vertex~$x$ to vertex~$x\cdot y$. Furthermore, there is a double arrow linking every
idempotent~$x$ to its stabilised version~$x^\sharp$.
\end{example}

\begin{example}\label{example:stabilisation-monoid-sega}
Imagine we want to compute the size of the longest sequence of consecutive $a$'s
in words over the alphabet~$\{a,b\}$.
Then we would separate four ``kinds'' of words:
\begin{iteMize}{$\bullet$}
\item the kind consisting only of the empty word; let it be~$1$,
\item the kind of words, containing only occurrences of $a$,
	at least one occurrence of it, but only ``a few'' of them; let the corresponding element be~$a$,
\item the kind of words containing at least one~$b$, but no long sequence of consecutive~$a$'s; let the corresponding element be~$b$,
\item the kind of words that contain a long sequence of consecutive~$a$'s; let the corresponding element be~$0$.
\end{iteMize}
Computing the size of the longest sequence of consecutive $a$'s means
identifying the words containing a ``long'' sequence of this type, {\it i.e.}, it means to separate
words of kind $0$ from words of kind $a$ or $b$.

The table of product and stabilisation is then naturally the following:
\begin{align*}
\begin{array}{c|ccccc|c}
\cdot	&~1~&~a~&~b~&~0~&~~~&~\sharp~\\
\hline
1~	&1&a&b&0&&1\\
a	&a&a&b&0&&0\\
b	&b&b&b&0&&b\\
0	&0&0&0&0&&0
\end{array}
\end{align*}

We complete the definition of this stabilisation monoid by defining the ordering.
For this, we let~$x\leq x$ hold for all~$x\in\{1,a,b,0\}$, and we further set~$0\leq a$ since~$a^\sharp=0$. 
Since $0 = 0\cdot b$, $0\leq a$ and~$a\cdot b=b$, we need also to set~$0\leq b$
for ensuring the compatibility of the product with the order.
Once more the ordering corresponds to the intuition that there exist words that can be of kind~$a$ (e.g., the word $a^{100}$)
or~$b$ (e.g., the word $ba^{100}$), and that have kind~$0$ if we change what we mean by ``a lot of''.
\end{example}

\begin{remark}\label{remark:standard-monoid}
The notion of stabilisation monoids (or stabilisation semigroups)
extends the notion of standard monoids (or semigroups).
Many standard results concerning monoids have natural counterparts
in the world of stabilisation monoids. For making this relationship more precise,
let us describe the canonical  way to translate a monoid into a stabilisation monoid.
Let~$\monoid=\langle M,\cdot\rangle$ be a monoid. The corresponding stabilisation monoid is:
\begin{align*}
\monoid_\sharp=\langle M, \cdot, =, \mathit{id}_{E(M)}\rangle\ .
\end{align*}
In other words, the monoid is extended with a trivial ordering (the equality),
and the stabilisation is simply the identity over idempotents.
The reader can easily check that this object indeed respects the
definition of a stabilisation monoid.

If we refer to the intuition we gave above, we extend the monoid by an identity stabilisation.
This means that for all idempotents, we do not make the distinction between
iterating it ``a few times'' or ``a lot of times''. Said differently, we never have to count the
number of occurrences of the idempotents.  This is consistent with the principle that 
a standard monoid has no counting capabilities.
\end{remark}

\begin{remark}
The order plays an important role, even if it is sometimes hidden.
Let us first remark that given a stabilisation monoid, it may happen that
changing the order yields again another valid stabilisation monoid
(as for ordered monoids). In general, there is a least order such that the
structure is a valid stabilisation monoid. It is the intersection of all the
``valid orders'', and can be computed by a least fix-point. However,
there is no maximal ``valid order'' in general.

More interestingly, there exist structures $\langle M,\cdot,\sharp\rangle$
which have no order, which satisfy the definition of a stabilisation monoid,
excepting for the rules involving the order, and such that it is not possible to
construct an order for making them a valid stabilisation monoid.
An example is the 10 elements structure which would be obtained
for describing the property ``there is an even number of small maximal segments of $a$'s''.
But this is what we want. Indeed, a closer inspection would reveal that this
property does not have the monotonic behaviour that we could use for defining a function.
Consider for instance a word of the form $a^1ba^2ba^3b\dots ba^n$,
and assume a \intro{small maximal segment of consecutive $a$'s} means a segment of length at most $m$,
then, if $m$ takes an even values at most equal to $n$, the word should be considered as in the language,
while if it takes an odd value at most equal to $m$,  the word should be thought outside the language.
Thus, when $m$ ranges in the interval $\{0,\dots,n\}$, the word is alternatively thought as in the language or
outside the language. This is typically a non-monotonic behaviour.
Keeping in mind cost monadic logic from the
previous section, we see that no formula would be able to express such a property.
Requiring an order in the definition of stabilisation monoids 
rules out such situations.
\end{remark}

We have seen through the above examples how easy it is to work
with stabilisation monoids at an informal level. An important part of the rest of the section
is dedicated to providing formal definitions for this informal reasoning.
In the above explanations, we worked with the imprecise terminology
``a few'' and ``a lot of''. Of course, the value (what we referred to as ``the kind'' in the examples)
of a word depends on what is the frontier we fix for separating ``a few'' from ``a lot''.

We continue the description of stabilisation monoids by introducing the key notion
of computations. These objects describe how to evaluate a long ``product'' in
a stabilisation monoid.


\subsection{Computations, under-computations and over-computations}
\mylabel{subsection:computations}

Our goal is now to provide a formal meaning for the notion of
stabilisation semigroups and stabilisation monoids,
allowing to avoid terms such as ``a lot'' or ``a few''.
More precisely, we develop in this section the notion of computations.
A computation is a tree which is used as a witness that a word
evaluates to a given value.

We fix ourselves for the rest of the section a stabilisation semigroup
$\semigroup=\langle S,\cdot,\sharp,\leq\rangle$.
We develop first the notion for semigroups, and then see how to
use it for monoids in Section~\ref{subsection:monoid-computations}
(we will see that the notions are in close correspondence).

Let us consider a word $u\in S^+$ (it is a word over $S$, seen as an alphabet).
Our objective is to define a ``value'' for this word. 
In standard semigroups, the ``value'' of $u$ is simply~$\pi(u)$, the product of
the elements appearing in the word. But, what should the ``value'' be for a stabilisation
semigroup?
All the informal semantics we have seen so far were
based on the distinction between ``a few'' and
``a lot''. This means that the value the word has depends on what is
considered as ``a few'', and what is considered as ``a lot''. This is captured
by the fact that the value is parameterised by a positive integer~$n$
which can be understood as a \intro{threshold} separating what is considered
as ``a few'' from what is considered as ``a lot''. For each choice of~$n$,
the word~$u$ is subject to have a different value in the stabilisation
semigroup.


Let us assume a threshold value~$n$ is fixed. We still lack a general mechanism
for associating to each word~$u$ over~$S$ a value in~$S$.
This is the purpose of \intro{computations}.
Computations are proofs (taking the form of a tree) that a
word should evaluate to a given value.
Indeed, in the case of usual semigroups, the fact that a word~$u$ evaluates
to~$\pi(u)$ can be witnessed by a binary tree, the leaves of which, read
from left to right, yield the word~$u$, and such that each inner node is
labelled by the product of the label of its children. Clearly, the root
of such a tree is labelled by~$\pi(u)$, and the tree can be seen as
a proof of correctness for this value.

The notion of $n$-computation that we define now is a variation around
this principle. For more ease in its use, it comes in three variants:
under-computations, over-computations and computations.
\begin{figure}
\begin{center}
\includegraphics[scale=0.5]{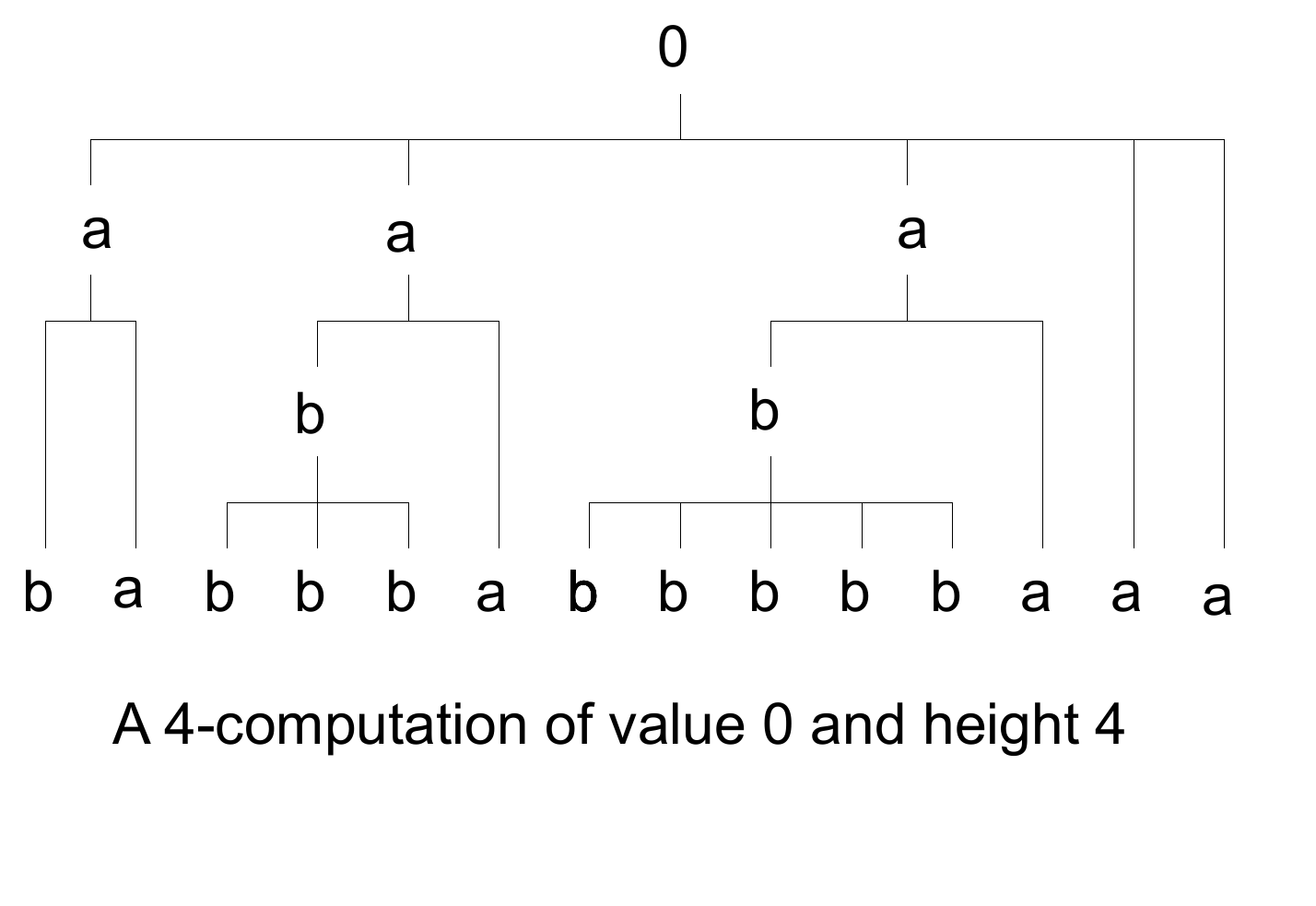}

\includegraphics[scale=0.5]{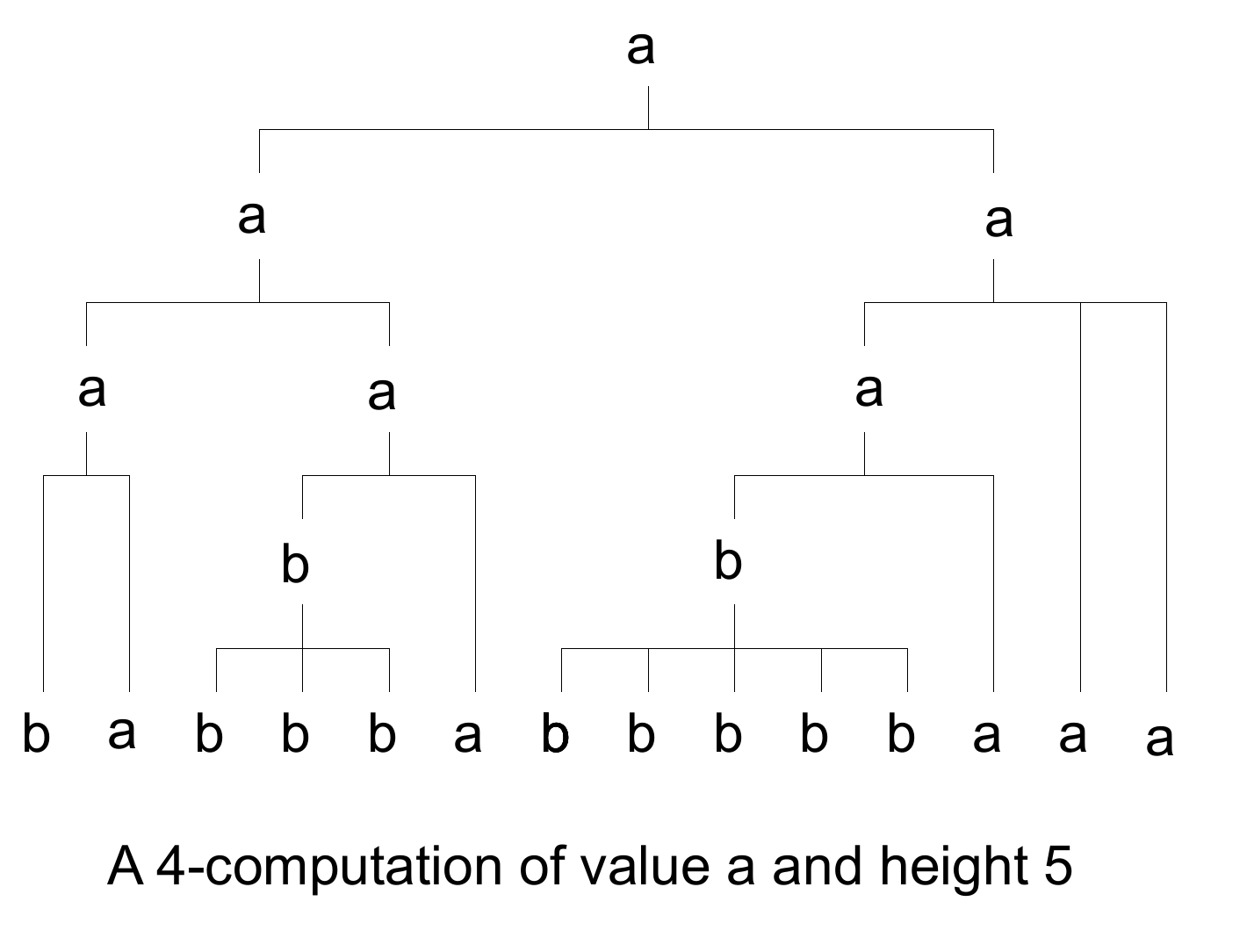}
\end{center}
\caption{Two $4$-computations for the word $\mathit{babbbabbbbbaaa}$ in the
monoid of Example~\ref{example:stabilisation-monoid-counta}.}
\mylabel{fig:calculs}
\end{figure}

\begin{definition} An \intro{$n$-under-computation~$T$ for the word~$u=a_1\dots a_l\in S^+$}
is an ordered unranked tree with $l$ leaves, each node~$x$ of which is labelled by
an element~$v(x)\in S$ called the \intro{value of~$x$}, and such that
for all nodes~$x$ of children $y_1,\dots,y_k$ (read from left to right), one
of the following cases holds:
\begin{description}
\item[Leaf] $k=0$, and $v(x)\leq a_m$ where $x$ is the $m$th leave of $T$ (read from left to right),
\item[Binary node] $k=2$, and~$v(x)\leq v(y_1)\cdot v(y_2)$,
\item[Idempotent node] $2\leq k\leq n$ and~$v(x)\leq e$ where $e=v(y_1)=\dots=v(y_k)\in E(M)$,
\item[Stabilisation node] $k>n$ and $v(x)\leq e^\sharp$ where $e=v(y_1)=\dots=v(y_k)\in E(M)$.
\end{description}
An \intro{$n$-over computation} is obtained by replacing everywhere ``$v(x)\leq$''
	by ``$v(x)\geq$''.
An \intro{$n$-computation} is obtained by replacing everywhere ``$v(x)\leq$'' by ``$v(x)=$'', {\it i.e.},
a $n$-computation is a tree which is at the same time
an $n$-under computation and an $n$-over computation.

The \intro{value} of a [under-/over-]computation is the value of its root.

We also use the following notations for easily denoting constructions of [under/over]-computations.
Given a non-leaf $S$-labelled tree $T$, denote by $T^i$ the subtree of $T$
rooted at the $i^\text{th}$ children of the root.
For $a$ in $S$, we note as $a$ the tree restricted
to a single leaf of value $a$. If furthermore $T_1,\dots,T_k$ are also $S$-labelled trees,
then $a[T_1,\dots,T_k]$ denotes the tree of root labelled $a$, of degree $k$,
and such that $T^i=T_i$ for all $i=1\dots k$.
\end{definition}

It should be immediately clear that these notions have
to be manipulated with care, as shown by the following example.
\begin{example} 
Two examples of computations are given in Figure~\ref{fig:calculs}.
Both correspond to the stabilisation semigroup (in fact monoid)
of Example~\ref{example:stabilisation-monoid-counta},
the aim of which is to count the number of occurrences of the letter~$a$ 
in a word. Both correspond to the evaluation of the same word.
Both correspond to the same threshold value~$n$.
However, these two computations do not have the same value.
We will see below how to compare computations and overcome
this problem.
\end{example}

There is another problem. Indeed, it is straightforward to construct an $n$-computation
for some word, simply by constructing a computation which is a binary tree, and
would use no idempotent nodes nor stabilisation nodes.
However, such a computation would of course not be satisfactory since 
every word $u$ would be evaluated in this way as $\pi(u)$.
We do not want that. This would mean that the quantitative aspect contained in the stabilisation has been lost.
We need to determine what is a relevant computation in order to rule out such computations.

Thus we need to answer the following questions:
\begin{enumerate}[(1)]
\item What are the relevant computations?
\item Can we construct a relevant $n$-computation for all words and all~$n$?
\item How do we relate the different values that $n$-computations may have on
	the same word?
\end{enumerate}

The answer to the first question is that we are only interested in computations
of small height, meaning of height bounded by some function of the semigroup.
With such a restriction, it is not possible to use binary trees as computations.
However, this choice makes the answer to the second question less obvious:
does there always exist a computation?
\begin{theorem}\mylabel{theorem:exists-computation}
For all words~$u\in S^+$ and all non-negative integers~$n$,
there exists an $n$-computation of height at
most\footnote{When measuring the height of a tree, the convention
	is that leaves do not count. With this convention, substituting a tree for a leaf of another tree
	results in a tree of height the sum of the heights of the two trees.} $3|S|$.
\end{theorem}
This result is an extension of the forest factorisation theorem of Simon~\cite{Simon90}
(which corresponds to the case of a semigroup). Its proof,
which is independent from the  rest of this work, is presented in
Section~\ref{subsection:proof-exists-computation}.

\begin{figure}[ht]
\begin{center}
\null\hspace{-5cm}
\includegraphics[scale=1]{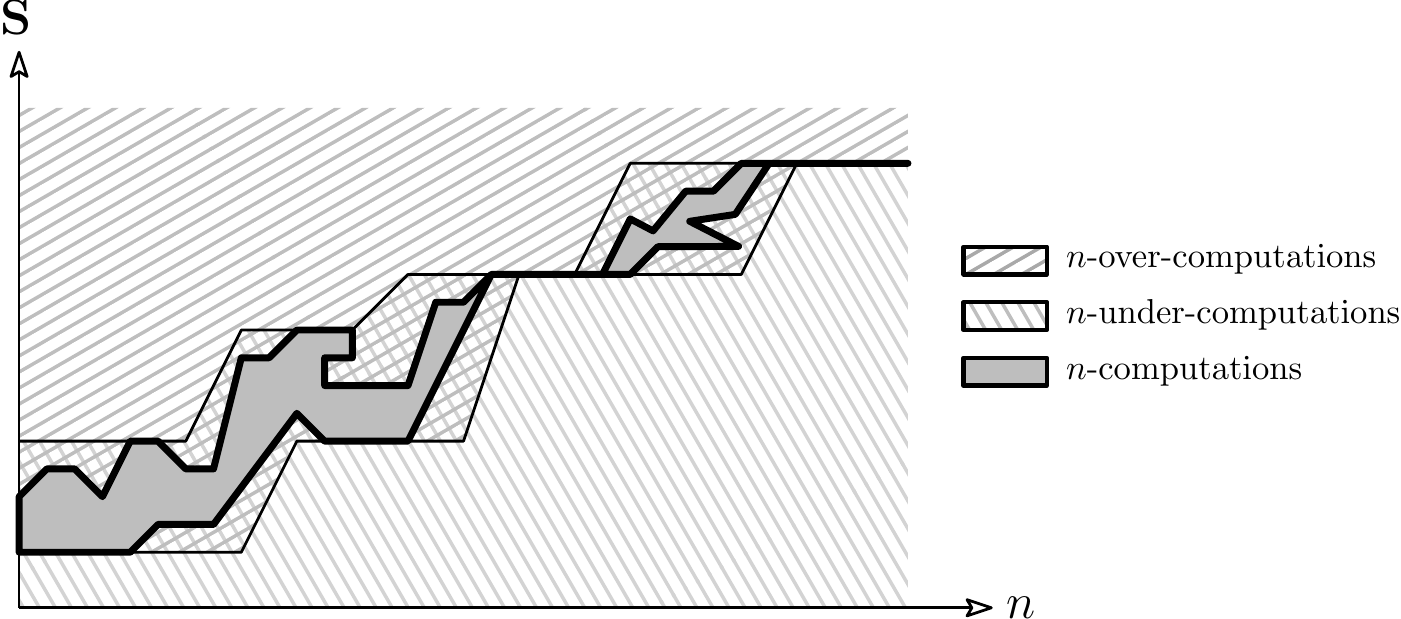}
\hspace{-5cm}\null
\end{center}
\caption{Structure of computations, under-computations, and over-computations.\mylabel{figure:structure-computations}}
\end{figure}

The third question remains: how to compare the values of different
computations over the same word? An answer to this question  in its full generality makes use
of under- and over-computations.
\begin{theorem}\mylabel{theorem:unicity-computations}
For all non-negative integers~$p$, there exists a polynomial~$\alpha:\nats\rightarrow\nats$
such that for all $n$-under-computations over some word of value~$a$ and
height at most~$p$, and all $\alpha(n)$-over computations of value~$b$
over the same word~$u$,
$$a\leq b\ .$$ 
\end{theorem}
Remark first that since computations are special instances of under- and
over-compu\-tations, Theorem~\ref{theorem:unicity-computations}
holds in particular for comparing the values of compu\-tations.
The proof of Theorem~\ref{theorem:unicity-computations} is the subject of
Section~\ref{subsection:proof-unicity-computations}. 

We have illustrated the above results in Figure~\ref{figure:structure-computations}.
It depicts the relationship between computations in some idealised stabilisation monoid $\semigroup$.
In this drawing, assume some word over some stabilisation semigroup is fixed,
as well as some integer~$p\geq 3|\semigroup|$.
We aim at representing for each~$n$ the possible values of
an $n$-computation, an $n$-under computation or an $n$-over
computation for this word of height at most $p$.
In all the explanations below, all computations are supposed to
not exceed height~$p$.

The horizontal axis represents the $n$-coordinate.
The values in the stabilisation semigroup being ordered,
the vertical axis represents the values in the stabilisation semigroup (for the picture, we
assume the values in the stabilisation semigroup totally ordered).
Thus an $n$-computation (or $n$-under or $n$-over-computation) is placed
at a point of horizontal coordinate $n$ and vertical coordinate the value
of the computation.

We can now interpret the properties of the computations in terms of this figure.
First of all, under-computations as well as over-computations, and as opposed to
computations, enjoy certain forms of monotonicity
as shown by the fact below.
\begin{myfact}\mylabel{fact:under-over-computation-monotonicity}
For~$m\leq n$, all $m$-under-computations are also $n$-under-computations,
and all $n$-over-computations are also $m$-over-computations (using the fact that $e^\sharp\leq e$
for all idempotents $e$).

Any $n$-under-computation of value~$a$ can be turned into an $n$-under-computation
of value $b$ for all $b\leq a$ (by changing the root label from $a$ to $b$).
Similarly any $n$-over-computation of value~$a$ can be turned into an
$n$-over-computation of value $b$ for all $b\geq a$.
\end{myfact}
Fact~\ref{fact:under-over-computation-monotonicity} is illustrated
by Figure~\ref{figure:structure-computations}.
It means that over-computations define a left and upward-closed area, while the under-computations
define a right and downward-closed area. Hence, in particular, the delimiting lines
are non-decreasing.
Furthermore, since computations are at the same-time
over-computations and under-computations, the area of computations lie inside
the intersection of under-compu\-ta\-tions and over-com\-pu\-ta\-tions. 
Since the height~$p$ is chosen to be at least $3|\semigroup|$,
Theorem~\ref{theorem:exists-computation} provides for us even more information.
Namely, for each value of $n$, there exists an $n$-computation.
This means in the picture that the area of computations crosses every column.
However, since computations do not enjoy monotonicity properties, the shape
of the area of computations can be quite complicated.
Finally Theorem~\ref{theorem:unicity-computations} states that the frontier
of under-computations and the frontier of over-computations are not far one
from each other. More precisely, if we choose an element $a$
of the stabilisation semigroup, and we draw an horizontal line at altitude~$a$,
if the frontier of under-computations is above or at $a$ for threshold $n$,
then the frontier of over-computations is also above or at $a$
at threshold $\alpha(n)$. Hence the frontier of over-computations is always below
the one of under-computations, but it essentially grows at the same
speed, with a delay of at most $\alpha$.

\begin{remark}\mylabel{remark:computations-standard}
In the case of standard semigroups or monoids (which can be seen as
stabilisation monoids or semigroups according to Remark~\ref{remark:standard-monoid}),
the notions of computations, under-computations and
over-computations coincide (since the order is trivial),
and the value of the threshold $n$ becomes also
irrelevant. This means that the value of all $n$-[under/over-]computations
over a word~$u$ coincide with $\pi(u)$. (Such computations coincide with
the ``Ramsey factorisations'' of the factorisation forest theorem.)
\end{remark}

Let us finally remark that Theorem~\ref{theorem:unicity-computations}, which is a consequence
of the axioms of stabilisation semigroups, is also sufficient for deducing them.
This is formalised by the following proposition.
\begin{proposition}\mylabel{proposition:unicity-to-axioms}
Let $\semigroup=\langle S,\cdot,\leq,\sharp\rangle$ be a a structure consisting of a finite set $S$,
a binary operation $\cdot$ from $S^2$ to $S$, $\leq$ be a partial order, and $\sharp$
from $S$ to $S$ be a mapping defined over the idempotents of $S$.
Assume furthermore that there exists $\alpha$ such that for all $n$-under-computations for
some word~$u$ of value $a$ of height at most $3$
and all $\alpha(n)$-over-computation over $u$ of
value $b$ of height at most $3$,
$$a\leq b\ .$$

Then $\semigroup$ is a stabilisation semigroup.
\end{proposition}
\begin{proof}
Let us first prove that $\cdot$ is compatible with $\leq$.
Assume $a\leq a'$ and $b\leq b'$, then $(a\cdot b)[a,b]$ is a $0$-under-computation
over $ab$ and $(a'\cdot b')[a',b']$ is an $\alpha(0)$-over-computation over the same
word $ab$. It follows that $a\cdot b\leq a'\cdot b'$.

Let us now prove that $\cdot$ is associative.
Let $a,b,c$ in $n$. Then $((a\cdot b)\cdot c)[(a\cdot b)[a,b],c]$ is a $0$-computation
for the word $abc$, and $(a\cdot(b\cdot c))[a,(b\cdot c)[b,c]]$ is an $\alpha(0)$-computation
for the same word. It follows that $(a\cdot b)\cdot c\leq a\cdot(b\cdot c)$.
The other inequality is symmetric.

Let $e$ be an idempotent. The tree $e^\sharp[\overbrace{e,\dots,e}^{2\alpha(0)+2}]$
is both a $0$ and $\alpha(0)$-computation over the word $e^{2\alpha(0)+1}$.
Furthermore, the tree $(e^\sharp\cdot e^\sharp)[e^\sharp[\overbrace{e,\dots,e}^{\alpha(0)+1}],e^\sharp[\overbrace{e,\dots,e}^{\alpha(0)+1}]]$ is also both a $0$ and an $\alpha(0)$-computation for the same
word. It follows that $e^\sharp\cdot e^\sharp=e^\sharp$, {\it i.e.}, that $\sharp$ maps idempotents to
idempotents.

Let us show that $e^\sharp\leq e$ for all idempotents~$e$.
The tree $e^\sharp[e,e,e]$ is a $2$-computation over the word $eee$,
and $e[e,e,e]$ is a $max(3,\alpha(2))$-computation over the same word.
It follows that $e^\sharp\leq e$.

Let us show that stabilisation is compatible with the order.
Let $e\leq f$ be idempotents.
Then $e^\sharp[\overbrace{e,\dots,e}^{\alpha(0)+1}]$
and $f^\sharp[\overbrace{f,\dots,f}^{\alpha(0)+1}]$ are respectively
a $0$-computation for the word $e^{\alpha(0)+1}$ and 
an $\alpha(0)$-over-computation for the same word. It follows that $e^\sharp\leq f^\sharp$.

Let us prove that stabilisation is idempotent.
Let $e$ be an idempotent. We already know that $(e^\sharp)^\sharp\leq e^\sharp$ (this makes sense since we have
seen that $e^\sharp$ is idempotent). Let
us prove the opposite inequality. 
Consider the $0$-computation $e^\sharp[\overbrace{e,\dots,e}^{(\alpha(0)+1)^2}]$ for the
word $e^{(\alpha(0)+1)^2}$, and the $\alpha(0)$-computation
$(e^\sharp)^\sharp[\overbrace{e^\sharp[\overbrace{e,\dots,e}^{\alpha(0)+1}],
	\dots,e^\sharp[\overbrace{e,\dots,e}^{\alpha(0)+1}]}^{\alpha(0)+1}]$
for the same word. It follows that $e^\sharp\leq (e^\sharp)^\sharp$.

Let us finally prove the consistency of stabilisation.
Assume that both $a\cdot b$ and $b\cdot a$ are idempotents.
Let $t_{ab}$ be $(a\cdot b)[a,b]$ (and similarly for $t_{ba}$), {\it i.e.}, computations for $ab$ and $ba$
respectively.
Define now:
\begin{align*}
	t_{(a\cdot b)^\sharp}&=
		(a\cdot b)^\sharp[\overbrace{t_{ab},\dots,t_{ab}}^{\alpha(0)+2~\text{times}}]\ ,\\
\text{and} \qquad
t_{a\cdot (b\cdot a)^\sharp\cdot b}&=(a\cdot (b\cdot a)^\sharp\cdot b)[
		a,((b\cdot a)^\sharp\cdot b)[
			(b\cdot a)^\sharp[\overbrace{t_{ba},\dots,t_{ba}}^{\alpha(0)+1~\text{times}}]
			,b]]\ .\\
\end{align*}
Then both $t_{(a\cdot b)^\sharp}$ and $t_{a\cdot (b\cdot a)^\sharp\cdot b}$
are at the same time $0$ and $\alpha(0)$-computations over the same word $(ab)^{\alpha(0)+2}$
of height at most~$3$.
Since their respective values are $(a\cdot b)^\sharp$ and $a\cdot (b\cdot a)^\sharp\cdot b$,
it follows by our assumption that $(a\cdot b)^\sharp=a\cdot (b\cdot a)^\sharp\cdot b$.
\end{proof}
This result is particularly useful. Indeed, when constructing a new stabilisation semigroup,
we usually aim at proving that it ``recognises'' some function (to be defined in the next chapter). 
It involves proving the hypothesis of Proposition~\ref{proposition:unicity-to-axioms}. Thanks to
Proposition~\ref{proposition:unicity-to-axioms}, the syntactic correctness is then for free.
This situation occurs in particular in Section~\ref{subsection:closure-inf-projection}
and \ref{subsection:closure-sup-projection} when the closure of recognisable
cost-functions under inf-projection and sup-projection is established.

\subsection{Specificities of stabilisation monoids}
\mylabel{subsection:monoid-computations}

We have presented so far the notion of computations in the case of stabilisation semigroups.
We are in fact interested in the study of stabilisation monoids.
Monoids differ from semigroups by the presence of a unit element~$1$.
This element is used for modelling the empty word.
We present in this section the natural variant of the notions of computations
for the case of stabilisation monoids. As is often the case, results from stabilisation
semigroups transfer naturally to stabilisation monoids.
The definition is highly related to the one for stabilisation semigroups,
and we see through this section that it is easy to go from the notion
for stabilisation monoid to the one of stabilisation semigroup case, and backward.
The result is that we use the same name ``computation'' for the two notions elsewhere
in the paper.
\begin{definition}
Let $\monoid$ be a stabilisation monoid.
Given a word $u\in M^*$, a \intro{stabilisation monoid $n$-[under/over]-computation} (\intro{sm}-[under/over]-computation
for short) for $u$
is an $n$-[under/over]-computation for some $v\in M^+$,
such that it is possible to obtain $u$ from $v$ by deleting some occurrences of the
letter $1$. All have value $1$.
\end{definition}
Thus, the definition deals with the implicit presence of arbitrary many copies of the
empty word (the unit) interleaved with a given word.
This definition allows us to work in a transparent way with the empty word (this saves us
case distinctions in proofs). In particular the empty word has an 
sm-$n$-computation which is simply $1$, of value $1$. 
There are many others, like $1[1,1,1,1[1,1]]$ for instance.

Since each $n$-computation is also an sm-n-computation over the same word,
it is clear that Theorem \ref{theorem:exists-computation} can be extended to this situation (just the
obvious case of the empty word needs to be treated separately):
\begin{myfact}\mylabel{fact:exists-SM-computation}
There exists an sm-$n$-computation for all words in $M^*$
of size at most $3|M|$.
 \end{myfact}

The following lemma shows that sm-[under/over]-computations are not more expressive than
[under/over]-computations. It is also elementary to prove.
\begin{lemma}\mylabel{lemma:remove-epsilon-computation}
Given an sm-$n$-computation
({\it resp.}, sm-$n$-under-computation, sm-$n$-over-com\-putation)
of value $a$ for the empty word,  then $a=1$ ({\it resp.} $a\leq 1$, $a\geq 1$).

For all non-empty words $u$ and all sm-$n$-computations $T$
({\it resp.}, sm-$n$-under-computa\-tions, sm-$n$-over-computations)
for $u$ of value $a$, there exists an  $n$-computation
({\it resp.}, $n$-under-computations, $n$-over-computations) for $u$ of value $a$. Furthermore, its height is at most the height of $T$.
\end{lemma}
\begin{proof}
It is simple to eliminate each occurrence of an extra $1$ by local modifications of the structure of the sm-computation:
replace subtrees of the form $1[1,\dots,1]$ by $1$, subtrees of the form $a[T,1]$ by $T$, and subtrees of the form  $a[1,T]$ by $T$,
up to elimination of all occurrences of $1$.
For the empty word, this results in the first part of the lemma. For non-empty words, the resulting simplified
sm-computation is a computation. The argument works identically for the under/over variants.
\end{proof}
A corollary is that Theorem~\ref{theorem:unicity-computations} extends to sm-computations.
\begin{corollary}
For all non-negative integers~$p$, there exists a polynomial~$\alpha:\nats\rightarrow\nats$
such that for all sm-$n$-under-computations over some word $u\in M^*$ of value~$a$ and
height at most~$p$, and all sm-$\alpha(n)$-over computations of value~$b$
over the same word~$u$,
$$a\leq b\ .$$ 
\end{corollary}
\begin{proof}
Indeed, the sm-under-computations and sm-over-computations can be turned into
under-computations and over-computations of the same respective values by
Lemma~\ref{lemma:remove-epsilon-computation}.
The inequality holds for these under and over-computations by  Theorem~\ref{theorem:unicity-computations}.
\end{proof}

There is a last lemma which is related and will prove useful.
\begin{lemma}\mylabel{lemma:SM-extension}
Let $u$ be a word in $M^*$ and $v$ be obtained from $u$ by eliminating some of its $1$ letters,
then all $n$-[under/over]-computations for $v$ can be turned into an $n$-[under/over]-computation for $u$
of same value. Furthermore, the height increase is at most $3$.
\end{lemma}
\begin{proof}
Let $v=a_1\dots a_n$, then $u=u_1\dots u_n$ for $u_i\in 1^*a_i1^*$.
Let $T$ be the $n$-[under/over]-computation for $v$ of value $a$.
It is easy to construct an $n$-[under/over]-computation $T_i$ for $u_i$ of height at most $3$
of value $a_i$. It is then sufficient to plug in $T$ each $T_i$ for the $i$th leave of $T$.
\end{proof}

The consequence of these results is that we can work with sm-[under/over]-computations
as with [under-over]-computations.
For this reason we shall not distinguish further between the two notions unless necessary.

\subsection{Existence of computations: the proof of Theorem~\ref{theorem:exists-computation}}
\mylabel{subsection:proof-exists-computation}

In this section, we establish Theorem~\ref{theorem:exists-computation} which states
that for all words~$u$ over a stabilisation semigroup~$\semigroup$ and all non-negative integers~$n$,
there exists an $n$-computation for~$u$ of height at most~$3|S|$.
Remark that the convention in this context is to measure the height of a tree without counting the leaves.
This result is a form of extension of the factorisation forest theorem
due to Simon \cite{Simon90}:
\begin{theorem}[Simon \cite{Simon90,Simon92}]\mylabel{theorem:factorisation-forest}
Define a \intro{Ramsey factorisation} to be an $n$-computation in the pathological case $n=\infty$
({\it i.e.}, there are no stabilisation nodes, and idempotent nodes are allowed to have arbitrary degree).

\noindent
For all non-empty words~$u$ over a finite semigroup~$\semigroup$, there exists a Ramsey factorisation
for~$u$ of height\footnote{The exact bound of $3|S|-1$ is due to Kufleitner \cite{Kufleitner08}.
It is likely that the same bound could be achieved for Theorem~\ref{theorem:exists-computation}.
We prefer here a simpler proof with a bound of $3|S|$.} at most~$3|\semigroup|-1$.
\end{theorem}
Some proofs of the factorisation forest theorem can be found in \cite{Kufleitner08,TCS10:colcombet-fct,LATA11:colcombet}.
Our proof could follow similar lines as the above one. Instead of that,
we try to reuse as much lemmas as possible from the above constructions.

For proving Theorem~\ref{theorem:exists-computation}, we will need one of Green's relations,
namely the $\gJ$-relation (while there are five relations in general).
Let us fix ourselves a semigroup~$\semigroup$.
We denote by $\semigroup^1$ the semigroup extended (if necessary)
with a neutral element $1$ (this transforms $\semigroup$ into a monoid).
Given two elements $a,b\in S$, $a \leq_\gJ b$ if~$a=x\cdot b\cdot y$ for some~$x,y\in S^1$.
If $a\leq_\gJ b$ and~$b\leq_\gJ b$, then~$a \gJ b$.
We write~$a<_\gJ b$ to denote~$a\leq_\gJ b$ and~$b\not\leq_\gJ a$.
The interested reader can see, {\it e.g.}, \cite{LATA11:colcombet}
for an introduction to the relations of Green (with a proof of the factorisation forest theorem),
or monographs such as~\cite{Lallement79}, \cite{Grillet95} or \cite{Pin86} for deep presentations
of this theory.
Finally, let us call a \intro{regular element}
in a semigroup an element~$a$ such that~$a\cdot x\cdot a=a$ for some~$x\in S^1$.

The next lemma gathers some classical results concerning finite semigroups.
\begin{lemma}\mylabel{lemma:J-regular}
Given a $\gJ$-class $J$ in a finite semigroup, the following facts are equivalent:
\begin{iteMize}{$\bullet$}
\item $J$ contains an idempotent,
\item $J$ contains a regular element,
\item there exist $a,b\in J$ such that $a\cdot b\in J$,
\item all elements in $J$ are regular,
\item all elements in $J$ can be written as $e\cdot c$ for some idempotent $e\in J$,
\item all elements in $J$ can be written as $c\cdot e$ for some idempotent $e\in J$.
\end{iteMize}
Such~$\gJ$-classes are called \intro{regular}. 
\end{lemma}
We will use the following technical lemma.
\begin{lemma}\mylabel{lemma:J-idempotent}
If $f=e\cdot x\cdot e$ for $e\gJ f$ two idempotents, then $e=f$.
\end{lemma}
\begin{proof}
We use some standard results concerning finite semigroups.
The interested reader can find the necessary material for instance in \cite{Pin86}.
Let us just recall that the relations $\leq_\gL$, $\leq_\gR$ and $\gL$ and $\gR$
are the one-sided variants of $\leq_\gJ$ and $\gJ$ ($\gL$ stands for ``left'' and $\gR$ for ``right'').
Namely, $a\leq_\gL b$ ({\it resp.} $a\leq_\gR b$) holds if $a=x\cdot b$ for some $x\in S^1$ ({\it resp.} $a=b\cdot x$),
and $\gL=\leq_\gL\cap\geq_\gL$ ({\it resp.} $\gR=\leq_\gR\cap\geq_\gR$).
Finally, $\gH=\gL\cap\gR$.

The proof is very short. By definition $f\leq_\gL e$ since $e\cdot x\cdot e=f$.
Since by assumption $f\gJ e$, we obtain $f\gL e$
(a classical result in \emph{finite} semigroups).
In a symmetric way $f\gR e$. Thus $f\gH e$. Since an $\gH$-class contains at most
one idempotent, $f=e$ (it is classical than any $\gH$-class, when containing an idempotent, has a group structure; since
groups contain exactly one idempotent element, this is the only one).
\end{proof}
The next lemma shows that the stabilisation operation behaves in a very uniform way
inside $\gJ$-classes (similar arguments can be found in the works of Leung, Simon and Kirsten).
\begin{lemma}\mylabel{lemma:stabilisation-preserve-J}
If $e\gJ f$ are idempotents, then $e^\sharp \gJ f^\sharp$. Furthermore, if $e=x\cdot f\cdot y$
for some $x,y$, then $e^\sharp=x\cdot f^\sharp\cdot y$.
\end{lemma}
\begin{proof}
For the second part,
assume $e=x\cdot f\cdot y$ and $e\gJ f$. Let $f'=(f\cdot y\cdot e\cdot x\cdot f)$.
We easily check $f'\cdot f'=f'$.
Furthermore $f\gJ e=(x\cdot f\cdot y)\cdot e\cdot  (x\cdot f\cdot y)\leq_\gJ f'\leq_\gJ f$. Hence $f\gJ f'$.
It follows by Lemma~\ref{lemma:J-idempotent} that $f'=f$.
We now compute $e^\sharp = (x\cdot f\cdot f\cdot y)^\sharp=x\cdot f\cdot(f\cdot x\cdot y\cdot f)^\sharp\cdot f\cdot y=x\cdot f\cdot f^\sharp\cdot f\cdot y=x\cdot f^\sharp\cdot y$ (using consistency and $f=f'$).

This proves that $e\gJ f$ implies $e^\sharp\leq_\gJ f^\sharp$. Using symmetry, we obtain $e^\sharp\gJ f^\sharp$.
\end{proof}
Hence, if $J$ is a regular $\gJ$-class, there exists a unique $\gJ$-class $J^\sharp$
which contains $e^\sharp$ for one/all idempotents $e\in J$. If $J=J^\sharp$,
then $J$ is called \intro{stable}, otherwise, it is called \intro{unstable}.
The following lemma shows that stabilisation is trivial over stable $\gJ$-classes.
\begin{lemma}\mylabel{lemma:stable}
If $J$ is a stable $\gJ$-class, then $e^\sharp=e$ for all idempotents $e\in J$.
\end{lemma}
\begin{proof}
Indeed, we have $e^\sharp=e\cdot e^\sharp\cdot e$ and thus by Lemma~\ref{lemma:J-idempotent}, $e^\sharp=e$. 
\end{proof}
The situation is different for unstable $\gJ$-classes.
In this case, the stabilisation always goes down in the $\gJ$-order.
\begin{lemma}\mylabel{lemma:unstable}
If $J$ is an unstable $\gJ$-class, then $e^\sharp<_\gJ e$ for all idempotents $e\in J$.
\end{lemma}
\begin{proof}
Since $e^\sharp=e\cdot e^\sharp$, it is always the
case that $e^\sharp\leq_\gJ e$. Assuming $J$ is unstable means that $e\gJ e^\sharp$
does not hold, which in turn implies $e^\sharp<_\gJ e$.
\end{proof}

We say that a word~$u=a_1\dots a_n$ in $S^+$
is \intro{$J$-smooth}, for $J$ a $\gJ$-class,
if $u\in J^+$, and $\pi(u)\in J$.
It is equivalent to say that  $\pi(a_ia_{i+1}\cdots a_j)\in J$ for all $1\leq i<j\leq n$.
Indeed for all $1\leq i<j\leq n$, $a_i\gJ \pi(a_1\dots a_n)\leq_\gJ \pi(a_ia_{i+1}\cdots a_j)\leq_\gJ a_i\in J$.
Remark that, according to Lemma~\ref{lemma:J-regular}, if $J$ is irregular,
$J$-smooth words have length at most $1$. 
We will use the following lemma from \cite{LATA11:colcombet} as a black-box.
This is an instance of the factorisation forest theorem, but restricted to
a single $\gJ$-class.
\begin{lemma}[{Lemma~14 in \cite{LATA11:colcombet}}]\mylabel{lemma:factorisation-J}
Given a finite semigroup $\semigroup$, one of its $\gJ$-classes $J$, and a $J$-smooth word~$u$,
there exists a Ramsey factorisation for $u$ of height at most $3|J|-1$. 
\end{lemma}

Remark that Ramsey factorisations and $n$-computations do only differ on what is allowed for a node
of large degree, {\it i.e.}, above $n$. That is why our construction makes use of Lemma~\ref{lemma:factorisation-J}
to produce Ramsey factorisations, and then based on the presence of nodes of large degree, constructs
a computation by gluing pieces of Ramsey factorisations together.
\begin{lemma}\mylabel{lemma:computation-base}
Let $J$ be a $\gJ$-class, $u$ be a $J$-smooth word, and $n$ be some non-negative integer.
Then one of the two following items holds:
\begin{enumerate}[(1)]
\item there exists an $n$-computation for $u$ of value $\pi(u)$ and height at most $3|J|-1$, or;
\item there exists an $n$-computation for some non-empty prefix $w$ of $u$ of value\footnote{A closer
		inspection would reveal that $a\in J^\sharp$. This extra information is useless for our purpose.}
	$a<_\gJ J$ and height at most $3|J|$.
\end{enumerate}
\end{lemma}
\begin{proof}
Remark that if $J$ is irregular, then $u$ has length $1$ by Lemma~\ref{lemma:J-regular},
and the result is straightforward.
Remark also that if $J$ is stable, and since the stabilisation is trivial in stable $\gJ$-classes
(Lemma~\ref{lemma:stable}),
every Ramsey factorisations for $u$ of height at most
$3|J|-1$ (which exist by Lemma~\ref{lemma:factorisation-J}) is in fact $n$-computations for $u$.

The case of $J$ unstable remains.
Let us say that a node in a factorisation is \intro{big} if its degree is more than $n$.
Our goal is to ``correct'' the value of big nodes.
If there is a Ramsey factorisation for $u$ which has no big node, then it can be seen
as an $n$-computation, and once more the first conclusion of the lemma holds.

Otherwise, consider the least non-empty prefix $u'$ of $u$ for which there is a  Ramsey factorisation of height 
at most $3|J|-1$ which contains a big node. Let $F$ be such a factorisation and $x$ be a big node in $F$
which is maximal for the descendant relation (there are no other big nodes below).
Let $F'$ be the subtree of $F$ rooted in $x$.
This decomposes $u'$ into $vv'v''$ where $v'$ is the factor of $u'$ for which $F'$ is a Ramsey factorisation.
For this $v'$, it is easy to transform $F'$ into an $n$-computation $T'$ for $v'$: just replace the label $e$ of the
root of $F'$ by $e^\sharp$. Indeed, since there are no other big nodes in $F'$ than the root, the root is the only place
which prevents $F'$ from being an $n$-computation. Remark that from Lemma~\ref{lemma:unstable},
the value of $F'$ is $<_\gJ J$.

If $v$ is empty, then $v'$ is a prefix of $u$, and $F'$ an $n$-computation for it. The second conclusion of the lemma holds.

Otherwise, by the minimality assumption and Lemma~\ref{lemma:factorisation-J},
there exists a Ramsey factorisation $T$ for $v$ of height at most $3|J|-1$ which contains no
big node. Both $T$ and $T'$ being $n$-computations of height at most $3|J|-1$,
it is easy to combine them into an $n$-computation of height at most $3|J|$ for $vv'$.
This is an $n$-computation for $vv'$, which inherits from $F'$ the property that its value
is $<_\gJ J$. It proves that the second conclusion of the lemma holds.
\end{proof}

We are now ready to establish Theorem~\ref{theorem:exists-computation}.
 \begin{proof}
 The proof is by induction on the size of a left-right-ideal $Z\subseteq S$, {\it i.e.}, 
$S^1\cdot Z\cdot S^1\subseteq Z$ (remark that a left-right-ideal is a union of $\gJ$-classes).
We establish by induction on the size of $Z$ the following induction hypothesis:
\begin{center}
\emph{IH:} for all words $u\in Z^++Z^*S$ there exists an $n$-computation
of height at most $3|Z|$ for $u$.
 \end{center}
 Of course, for $Z=S$, this proves  Theorem~\ref{theorem:exists-computation}.
 
The base case is when $Z$ is empty, then $u$ has length $1$, and a single node tree
establish the first conclusion of the induction hypothesis 
(recall that the convention is that the leaves do not count in the height,
and as a consequence a single node tree has height $0$).

Otherwise, assume $Z$ non-empty. There exists a maximal $\gJ$-class $J$
(maximal for $\leq_\gJ$) included in $Z$.
From the maximality assumption, we can check that $Z'=Z\setminus J$ is again a left-right-ideal.
Remark also that since $Z$ is a left-right-ideal, it is downward closed for $\leq_\gJ$.
This means in particular that every element $a$ such that $a<_\gJ J$ belongs to $Z'$.

\emph{Claim:} We claim ($\star$) that for all words $u\in Z^++Z^*S$,
\begin{enumerate}[(1)]
\item either there exists an $n$-computation of height $3|J|$ for $u$, or;
\item there exists an $n$-computation of height at most $3|J|$ for some non-empty prefix
of $u$ of value in $Z'$.
\end{enumerate}

Let $w$ be the longest $J$-smooth prefix of $u$. If there exists no such non-empty prefix, this
means that the first letter $a$  of $u$ does not belong to $J$.
Two subcases can happen.
If $u$ has length $1$, this means that $u=a$, and thus $a$ is an $n$-computation witnessing
the first conclusion of ($\star$). Otherwise $u$ has length at least $2$, and thus $a$ belongs to $Z$.
Since furthermore it does not belong to $J$, it belongs to $Z'$. In this case, $a$ is an $n$-computation
witnessing the second conclusion of ($\star$).

Otherwise, according to Lemma~\ref{lemma:computation-base} applied to $w$, two situations can occur.
The first case is when there is an $n$-computation $T$ for $w$ of value $\pi(w)$ and height
at most $3|J|-1$.
There are several sub-cases.  If $u=w$, of course, the $n$-computation $T$ is a witness
that the first conclusion of ($\star$) holds.
Otherwise, there is a letter $a$ such that $wa$ is a prefix of $u$. 
If $wa=u$, then $\pi(wa)[T,a]$ is an $n$-computation for $wa$ of height at most
$3|J|$, witnessing that the first conclusion of ($\star$) holds.
Otherwise, $a$ has to belong to $Z$ (because all letters of $u$
have to belong to $Z$ except possibly the last one).
But, by maximality of $w$ as a $J$-smooth prefix,
either $a\in Z'$, or $\pi(wa)\in Z'$. Since $Z'$ is a left-right-ideal, $a\in Z'$
implies $\pi(wa)\in Z'$. Then,  $\pi(wa)[T,a]$ is an $n$-computation
for $wa$ of height at most $3|J|$ and value $\pi(wa)\in Z'$. This time,
the second conclusion of ($\star$) holds.

The second case according to Lemma~\ref{lemma:computation-base}
is when there exists a prefix $v$ of $w$
for which there is an $n$-computation of height at most $3|J|$ of value $<_\gJ J$.
In this case, $v$ is also a prefix of $u$, and the value of this computation is in $Z'$.
Once more the second conclusion of ($\star$) holds.
This concludes the proof of Claim~($\star$).
\smallskip

As long as the second conclusion of the claim ($\star$)
applied on the word $u$ holds, this decomposes $u$ into $v_1u'$,
and we can proceed with $u'$.
In the end, we obtain that all words $u\in Z^++Z^*S$ can
be decomposed into $u_1\dots u_k$ such that
there exist $n$-computations $T_1,\dots,T_k$ of height at most $3|J|$
for $u_1,\dots,u_k$ respectively, and such that the values of $T_1,\dots,T_{k-1}$
all belong to $Z'$ (but not necessarily the value of $T_k$).
Let $a_1,\dots,a_k$ be the values of $T_1,\dots,T_k$ respectively.
The word $a_1\dots a_k$ belongs to $Z'^++Z'^*S$.
Let us apply the induction hypothesis to the word $a_1\dots a_k$. We obtain
an $n$-computation $T$ for $a_1\dots a_k$ of height at most $3|Z'|$.
By simply substituting $T_1,\dots,T_k$ to the leaves of $T$, we obtain an
$n$-computation for $u$ of height at most $3|J|+3|Z'|=3|Z|$. (Remark once more
here that the convention is to not count the leaves in the height. Hence the height after
a substitution is bounded by the sum of the heights.)
\end{proof}

\subsection{Comparing computations: the proof of Theorem~\ref{theorem:unicity-computations}}
\mylabel{subsection:proof-unicity-computations}

We now establish the second key theorem for computations, namely
Theorem~\ref{theorem:unicity-computations} which states that the result of computations
is, in some sense, unique. The proof works by a case analysis on the possible
ways the over-computations and under-computations may overlap. 
We perform this proof for stabilisation monoids, thus using sm-computations.
More precisely, all statements take as input computations, and output
sm-computations, which can be then normalised into non-sm computations.
The result for stabilisation semigroup can be derived from it.
We fix ourselves from now on a stabilisation monoid $\monoid$.

\begin{lemma}\mylabel{lemma:over-computation-above-pi}
For all $n$-over-computations of value~$a$ over a word~$u\in M^*$ of length at most~$n$, $\pi(u)\leq a$.
\end{lemma}
\begin{proof}
By induction on the height of the over-computation, using the fact that an $n$-over-computation for a word
of length at most $n$ cannot contain a stabilisation node.
\end{proof}

\begin{lemma}\mylabel{lemma:over-computation-above-sharp}
For all $n$-over-computations of value~$b$
over a word~$b_1 \dots b_k$ ($k\geq 1$) such that~$e\leq b_i$ for all~$i$,
and $e$ is an idempotent, then~$e^\sharp\leq b$.
\end{lemma}
\begin{proof}
By induction on the height of the over-computation.
\end{proof}

A sequence of words $u_1,\dots,u_k$ is called a \intro{decomposition} of~$u$ if~$u=u_1\dots u_k$.
We say that a non-leaf [under/over]-computation~$T$ for a word~$u$ \intro{decomposes} $u$ into~$u_1$,\dots,$u_k$
if the subtree rooted at the $i$th child of the root is an [under/over]-computation for $u_i$, for all $i=1\dots k$.
Our proof will mainly make use of over-computations.
For this reason, we introduce the following terminology.

We say that a word $u\in M^*$ $n$-\intro{evaluates} to $a\in M$ if there exists an $n$-over-computation for $u$
of value $a$. We will also say that $u_1,\dots,u_k$ $n$-evaluate to $b_1,\dots,b_k$
if $u_i$ $n$-evaluates to $b_i$ for all $i=1\dots k$.

This notion is subject to elementary reasoning such as (a) $u$ $n$-evaluates to $\pi(u)$ or (b) if $u_1,\dots,u_k$
$n$-evaluate to $b_1,\dots,b_k$ and $b_1\dots b_k$ $n$-evaluates to $b$, then $u_1\dots u_k$
$n$-evaluates to $b$.

The core of the proof is contained in the following property:

\begin{lemma}\mylabel{lemma:over-computation-decomposition}
There exists a polynomial~$\alpha$ such that for all $u_1,\dots,u_k\in M^*$,
if $u_1\dots u_k$ $\alpha(n)$-evaluates to $b$ then 
$u_1,\dots,u_k$ $n$-evaluate $b_1,\dots,b_k$, and $b_1\dots b_k$ $n$-evaluates to $b$,
for some $b_1,\dots,b_k\in M$.

\ign{
for all $\alpha(n)$-over-computations~$T$ over a word~$u$ of value~$b$,
and all decompositions of $u$ into~$u_1,\dots,u_k$, then:
\begin{enumerate}[(1)]
\item there exist $n$-over-computations~$B_1$ \dots $B_k$ for $u_1$,\dots,$u_k$ respectively,
	of value $b_1$,\dots,$b_k$ respectively, and,
\item there exists an $n$-over-computation~$B$ over~$b_1\dots b_k$ of value~$b$.
\end{enumerate}}
\end{lemma}
From this result, we can deduce Theorem~\ref{theorem:unicity-computations} as follows.
\begin{proof}[Proof of Theorem~\ref{theorem:unicity-computations}]
Let~$\alpha$ be as in Lemma~\ref{lemma:over-computation-decomposition}.
Let~$\alpha_p$ be the $p$th composition of $\alpha$ with itself. 
Let~$U$ be an $n$-under-computation of height at most $p$ for some word $u$ of value $a$,
and $T$ be an $\alpha_p(n)$-over-computation for $u$ of value $b$. We want to establish that $a\leq b$.
The proof is by induction on~$p$.

If~$p=0$, this means that $u$ has length $1$, then $T$ and $U$ are also restricted to a single leaf, and the
result obviously holds. Otherwise, $U$ decomposes $u$ into $u_1,\dots,u_k$.
Let $a_1$,\dots,$a_k$ be the values of the children of the root of $i$, read from left to right.
By applying Lemma~\ref{lemma:over-computation-decomposition}
on $T$ and the decomposition $u_1,\dots,u_k$. We construct the
$\alpha_{p-1}(n)$-over-computations $B_1,\dots,B_k$ for $u_1,\dots,u_k$ respectively,
and of respective values $b_1,\dots,b_k$,
as well as an $\alpha_{p-1}(n)$-over-computation $B$ of value $b$ for $b_1,\dots,b_k$.

For all $i=1\dots k$, we can apply the induction hypothesis on
$U^i$ (let us recall that $U^i$ is the sub-under-computation rooted at the $i$th child of the root of $U$)
and $B_i$, and obtain that $a_i\leq b_i$. Depending on $k$, three cases
have to be separated.
If $k=2$ (binary node), then $a\leq a_1\cdot a_2\leq b_1\cdot b_2\leq b$.
If $3\leq k\leq n$ (idempotent node), we have $a\leq a_1=\cdots=a_k=e$ which is an idempotent.
We have $e=a_i\leq b_i$ for all $i=1\dots k$. Hence by Lemma~\ref{lemma:over-computation-above-pi},
$e\leq b$, which means $a\leq b$.
If $k>n$  (stabilisation node), we have once more $a_1=\cdots=a_k=e$ which is an idempotent, and 
such that $a\leq e^\sharp$.
This time, by Lemma~\ref{lemma:over-computation-above-sharp}, we have $e^\sharp\leq b$. We obtain
once more $a\leq b$.
\end{proof}

The remainder of this section is dedicated to the proof of Lemma~\ref{lemma:over-computation-decomposition}.

\begin{lemma}\mylabel{lemma:ordered-consistency}
There exists a positive integer $K$ such that for all idempotents $e,f$,
whenever $f\leq a_i\cdot b_i$ for all $i=1\dots k$, $k\geq K$,
and $b_i\cdot a_{i+1}\leq e$ for $i=1\dots k-1$, then $f^\sharp\leq a_1\cdot e^\sharp\cdot b_k$.
\end{lemma}
\begin{proof}
To each ordered pair $i<j$, let us associates the color $c_{i,j}=(a_i,\pi(b_ia_{i+1}\dots b_{j-1}))$.
We now apply the theorem of Ramsey to this coloring, for $K$ sufficiently large,
and get that there exist $1<i<s<j<k$ such that $c_{i,s}=c_{s,j}=_{c_{i,j}}=(a,b)$.
This implies in particular that $b\cdot a\cdot b=b$, and thus $a\cdot b$ and $b\cdot a$
are idempotents. Furthermore, $f\leq a\cdot b$ and $b\cdot a\leq e$.
It follows from consistency that
$f^\sharp\leq (a\cdot b)^\sharp\leq a\cdot (b\cdot a)^\sharp\cdot b\leq a\cdot e^\sharp\cdot b= a_i\cdot e^\sharp\cdot \pi(b_s\dots b_{j-1})$.
We now have, using the assumptions that $f\leq a_h\cdot b_h$ and $b_h\cdot a_{h+1}\leq e$,
\begin{align*}
f^\sharp	&=f\cdot f^\sharp\cdot f\leq \pi(a_1\dots a_i)\cdot e^\sharp\cdot \pi(b_s\dots b_k)
		\leq a_1\cdot e\cdot e^\sharp\cdot e\cdot b_k=a_1\cdot e^\sharp\cdot  b_k\ .
\end{align*}
\end{proof}

The following lemma will be used for treating the case
of idempotent and stabilisation nodes in the proof of
Lemma~\ref{lemma:over-computation-decomposition}.
\begin{lemma}\mylabel{lemma:conjugacy}
There exists a polynomial $\beta$ such that,
if $x_1,y_1,x_2,y_2\dots,x_m,y_m$ ($m\geq 1$) are elements of $M$
and $e$ is an idempotent
such that $x_h\cdot y_h\leq e$ for all $h=1\dots m$,
then $(y_1\cdot x_2)(y_2\cdot x_3)\cdots (y_{m-1}\cdot x_m)$
$n$-evaluates to $z$ such that:
\begin{iteMize}{$\bullet$}
\item $x_1\cdot z\cdot y_m\leq e$, and;
\item if $m>\beta(n)$ or $(x_h\cdot y_h)\leq e^\sharp$ for some $h=1\dots m$ then $x_1\cdot z\cdot y_m\leq e^\sharp$.
\end{iteMize}
\end{lemma}
\begin{proof}
Let us treat first the case $m\leq\beta(n)$, whatever is $\beta$.
Remark first that $(y_1\cdot x_2)\dots (y_{m-1}\cdot x_m)$
naturally $n$-evaluates to $z=\pi(y_1x_2\dots y_{m-1}x_m)$.
Thus $x_1\cdot z\cdot y_m=\pi(x_1y_1\cdots x_my_m)=e$.

Assume now that $m>\beta(n)$ for $\beta(n)=(n+K)^{3|M|}+1$
where $K$ is the constant obtained from Lemma~\ref{lemma:ordered-consistency}.
Set $d_i=(y_i\cdot x_{i+1})$ for all $i=1\dots m-1$.

We first claim ($\star$) that there exists $i<j$ such that $v=d_i\dots d_{j-1}$
$n$-evaluates to $y_i\cdot e^\sharp\cdot  x_j$.
For this, consider the word $u=d_1\dots d_{m-1}$,
and apply Theorem~\ref{theorem:exists-computation}
for producing an $(n+K)$-computation $U$ for $u$ of height at most $3|M|$.
The word $u$ has length $m-1>\beta(n)-1=(n+K)^{3|M|}$.
Thus there is a stabilisation node in $T$, say of degree $k>n+K$.
Let $S$ be a subtree of $T$ rooted at some stabilisation node.
Let $f$ be the (idempotent) value of the children of this node, the value of $S$ being $f^\sharp$.
This subtree corresponds to the factor $v=d_i\dots d_{j-1}$ of $u$.
We have to show that $v$ $n$-evaluates to $y_i\cdot e^\sharp\cdot x_j$.
The computation $S$ decomposes $v$ into $v_1,\dots, v_k$
and  each $v_h$ is of the form $d_{i_h}\dots d_{i_{h+1}-1}$
for some $h=1\dots k$ with $i=i_1<\dots<i_{k+1}=j$.
Define now $a_h$ to be $y_{i_h}$ and $b_h$ to be $\pi(x_{i_h+1} y_{i_h+1} \dots y_{i_{h+1}-1}x_{i_{h+1}})$
for all $h=1\dots k$. It is clear that $f\leq \pi(v_h)= a_h\cdot b_h$ for all $h=1\dots k$
since there is a computation over $v_h$ of value $f$.
Furthermore, $b_h\cdot a_{h+1}\leq e$ for all $h=1\dots k-1$.
Hence we can apply Lemma~\ref{lemma:ordered-consistency}, and get
that $f^\sharp\leq a_1\cdot e^\sharp\cdot b_k$.
Since furthermore $a_1=y_i$, and $b_k$ is either $\leq x_j$ or $\leq e\cdot x_j$,
it follows that $v$ $n$-evaluates to $f^\sharp\leq y_i\cdot e^\sharp\cdot x_j$.
This concludes the proof of the claim~($\star$).

Set now
$$z=\pi(y_1\dots y_i)\cdot e^\sharp\cdot \pi(x_j\dots y_m)\ .$$
Since, using the claim ($\star$), $d_1\dots d_{i-1}$, $d_i\dots d_{j-1}$, $d_j\dots d_{m-1}$
$n$-evaluate to $\pi(d_1\dots d_{i-1})$, $y_i\cdot e^\sharp\cdot x_j$,
	$\pi(d_j\dots d_{m-1})$, we get that $d_1\dots d_{m-1}$ $n$-evaluates to $z$.
Furthermore, $x_1\cdot z\cdot y_m=\pi(x_1\dots y_i)\cdot e^\sharp\cdot \pi(x_j\dots y_m)=e^\sharp$.
This proves the second conclusion of the statement.
\end{proof}

We are now ready to conclude.
\begin{proof}[Proof of Lemma~\ref{lemma:over-computation-decomposition}]
Let us set $\alpha(n)$ to be $(n+1)\beta(n)-1$,
where $\beta$ is the polynomial taken from Lemma~\ref{lemma:conjugacy}.
Lemma~\ref{lemma:over-computation-decomposition} follows from the following induction hypothesis:
\medskip

\noindent
{\em Induction hypothesis:} For all words
$u$ which $\alpha(n)$-evaluate to $b$,
and all decompositions of $u$ into $u_1,\dots,u_k$ ($k\geq 2$),
then $u_1,\dots,u_k$ $n$-evaluate to $b_1,\dots, b_k$,
and $b_2\dots b_{k-1}$ $n$-evaluate to $c$ such that $b_1\cdot c \cdot b_k\leq b$.

\noindent
{\em Induction parameter:} The height of the $\alpha(n)$-over-computation $T$
witnessing that  $\alpha(n)$-evaluates to $b$.
\medskip

It should be clear that this implies Lemma~\ref{lemma:over-computation-decomposition}
since this means that $b_1\dots b_k$ $n$-evaluate to $b_1\cdot c\cdot b_k\leq b$.

The essential idea in the proof of the induction hypothesis is that $T$
decomposes the word into $v_1,\dots,v_\ell$, and we have to study all
the possible ways the $v_i$'s and the $u_j$'s may overlap.
In practice, we will not refer much to $T$,
but simply about how it decomposes the word into $v_1,\dots,v_\ell$.
Thus, from now on, let $v_1,\dots,v_\ell$ and $u_1,\dots,u_k$ be decompositions of a word $u$
such that each of the $v_i$'s $\alpha(n)$-evaluates to $a_i$
and is subject to the application of the induction hypothesis.
\medskip 

\noindent
{\em Leaves.} This means that $\ell=1$. All the $u_h$'s should be empty, but one, say $u_h=a$
where $a$ is the letter labelling the leaf.
Three cases can occur depending on $h$. If $h=1$, then $u_1,\dots,u_k$ obviously $n$-evaluate
to $a,1,\dots,1,1$, and $1\dots1$ $n$-evaluate to $1$, and we indeed have $a\cdot 1\cdot 1\leq a$.
The case $h=k$ is symmetric. Finally, if $1<h<k$, then $u_1,\dots,u_k$ $n$-evaluate
to $1,\dots,1,a,1,\dots,1$, and $1\dots 1a1\dots 1$ $n$-evaluate to $a$, and we indeed have
$1\cdot a\cdot 1\leq a$.
\medskip

\noindent
{\em Binary nodes.}
If $\ell=2$, then there exist $s$ in $1,\dots, k$ and words $w,w'$ such that
\begin{align*}
v_1&=u_1\dots u_{s-1}w\ ,&
u_s&=ww'\ , &\text{and}\quad v_2&=w'u_{s+1}\dots u_k\ .
\end{align*}
We can apply the induction hypothesis to both $v_1$ and $v_2$.
We obtain that $u_1,\dots, u_{s-1},w$, $w',u_{s+1},\dots, u_k$ $n$-evaluate to
$b_1,\dots,b_{s-1},d,d',b_{s+1},\dots, b_k$,
and that $b_2\dots b_{s-1},b_{s+1}\dots b_k$ $n$-evaluate to $c_1,c_2$
such that $b_1\cdot c_1\cdot d\leq a_1$ and $d'\cdot c_2\cdot b_k\leq a_2$.
It follows that $u_s$ $n$-evaluates to $b_s=d\cdot d'$.
Furthermore, $b_2\dots b_{k-1}$ $n$-evaluates to $c_1\cdot b_s\cdot c_2=c$.
Overall, $u_1,\dots,u_k$ $n$-evaluate to $b_1,\dots,b_k$
and $b_2\dots b_{k-1}$ $n$-evaluates to $c$ such that $b_1\cdot c\cdot b_k=(b_1\cdot c_1\cdot d)\cdot(d'\cdot d_2\cdot b_k)\leq a_1\cdot a_2$.
\medskip

\noindent
{\em Idempotent and stabilisation nodes.}
Assume now that $v_1,\dots,v_\ell$ $\alpha(n)$-evaluate to $e,\dots,e$,
where $e$ is idempotent.
We aim at proving that $u_1,\dots,u_k$ $n$-evaluates to $b_1,\dots,b_k$,
and $b_2\dots b_{k-1}$ to $c$ such that $b_1\cdot c\cdot b_k\leq e$,
and if $\ell>\alpha(n)$, $b_1\cdot c\cdot b_k\leq e^\sharp$.

We rely on a suitable decomposition of the words:
there exist $0=i_0< i_1<\dots<i_m<i_{m+1}=\ell+1$ and $1=j_0<\dots<j_m=k$, as well as words
$\varepsilon=u'_0,u''_0,u'_1,u''_1,\dots,u'_m,u''_m=\varepsilon$ such that
\begin{align*}
v_{i_h}&=u''_{h-1}\,u_{j_{h-1}+1}\dots u_{j_{h}-1}\,u'_{h}&&
\text{for all $h=1\dots m$,}\tag{$\star$}\\
\text{and}\qquad
u_{j_h}&=u'_h\,v_{i_h+1}\dots v_{i_{h+1}-1}\,u''_h&&
\text{for all $h=0\dots m$.}\tag{$\star\star$}
\end{align*}
The best is to present it through a drawing. It is
annotated with all the variables that will be used during the proof.
The two main rows represent the two possible decompositions of the
word into $v_i$'s and $u_j$'s.
\begin{center}
\includegraphics[scale=1]{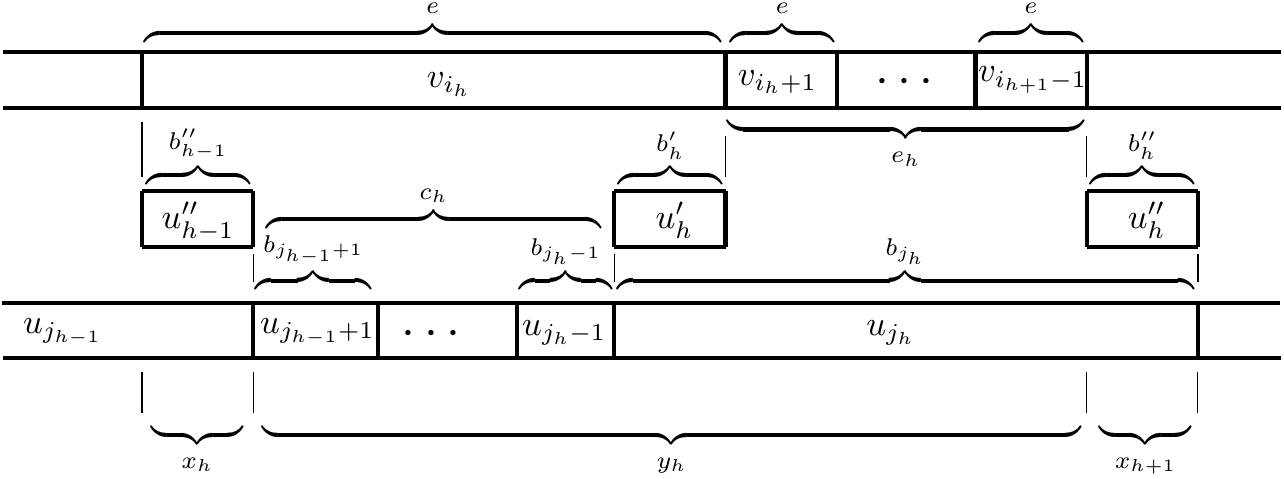}
\end{center}
Such a decomposition is not unique. It is sufficient to guarantee that each separation
between some $u_s$ and some $u_{s+1}$ fall in some $v_{i_h}$, and that
each $v_{i_h}$ contains such a separation.

We can apply the induction hypothesis on each equation $(\star\star)$.
Hence, it follows that
$u''_{h-1},u_{j_{h-1}+1},\dots, u_{j_{h}-1},u'_{h}$ $n$-evaluate to $b''_{h-1},b_{j_{h-1}+1},\dots, b_{j_{h}-1},b'_{h}$
and $b_{j_{h-1}+1} \dots b_{j_{h}-1}$ $n$-evaluates to $c_h$, such that
$b''_{h-1}\cdot c_h\cdot b'_{h}\leq e$.
Set furthermore $b'_0=b''_m=1$. We get that $u'_h,u''_h$ $n$-valuate to $b'_h,b''_h$ for all $h=0\dots m$.
Define furthermore for all $h=0\dots m$, $e_h$ as
\begin{align*}
e_h&=\begin{cases}
		1&\text{if}~i_{h+1}-i_h-1=0\\
		e&\text{if}~1\leq i_{h+1}-i_h-1\leq n\\
		e^\sharp&\text{if}~i_{h+1}-i_h-1> n
		\end{cases}
\end{align*}
Since each $v_h$ $\alpha(n)$-evauates to $e$, each $v_h$ also $n$-evaluates to $e$.
Now $e_h$ has been chosen such that $v_{i_h+1}\dots v_{i_{h+1}-1}$
$n$-evaluates to $e_h$. Thus from $(\star\star)$, $u_{j_h}$
$n$-evaluates for all $h=0\dots m$ to $b_{j_h}$ that we define
as $b_{j_h}=b'_h\cdot e_h\cdot b''_h$.
At this point, we have that
\begin{asparadesc}
\item[C1] $u_1,\dots,u_k$ $n$-evaluate to $b_1,\dots,b_k$.
\end{asparadesc}

To head toward the conclusion, we will use Lemma~\ref{lemma:conjugacy}.
Thus, let us set $x_{h}$ to be $b''_{h-1}$  and $y_{h}$ to be $c_h\cdot b'_h\cdot e_h$ for all $h=1\dots m$.
We have
\begin{align*}
x_h\cdot y_h&=(b''_{h-1}\cdot c_h\cdot b'_h)\cdot e_h=e\cdot e_h . \tag{$\dagger$}
\end{align*}
According to ($\dagger$), $x_h\cdot y_h\leq e$, and we can apply Lemma~\ref{lemma:conjugacy} to 
\begin{align*}
x_1,y_1,x_2,\dots, x_m,y_m
\end{align*}
and obtain that $(y_1\cdot x_2)\dots (y_{m-1}\cdot x_m )$ $n$-evaluates to some $z$
subject to the conclusions of the lemma (we will recall these conclusions upon need).

Let us now establish the following claims C2, C3 and C4.
\begin{asparadesc}
\item[C2] $b_2 \dots b_{k-1}$ $n$-evaluates to $(z\cdot c_m)$.\\
	Indeed, for all $h=1\dots m$, $c_h$ is chosen such that
	$b_{j_{h-1}+1} \dots b_{j_{h}-1}$ $n$-evaluates to $c_h$, thus
	$b_{j_{h-1}+1} \dots b_{j_{h}}$ $n$-evaluates to:
	\begin{align*}
	c_h\cdot b_{j_h}&=c_h\cdot b'_h\cdot e_h\cdot b''_h= y_h\cdot x_{h+1}\ ,
	\end{align*}
	by just unfolding the definitions. Since furthermore $(y_1\cdot x_2)\dots (y_{m-1}\cdot x_m )$ $n$-evaluates to $z$,
	it follows that $b_2\dots b_{j_{m-1}}$ $n$-evaluates to $z$.
	Furthermore, by choice of $c_m$, $b_{j_{m-1}+1}\dots b_{k-1}$ $n$-evaluates to $c_m$.
	Thus $b_2\dots b_{k-1}$ $n$-evaluates to $(z\cdot c_m)$ as claimed.
\item[C3] $b_1\cdot (z\cdot c_m) \cdot b_k\leq e$. \\
	Indeed, according to the conclusions of Lemma~\ref{lemma:conjugacy},
	$x_1\cdot z\cdot y_m\leq e$. Hence,
	\begin{align*}
	b_1\cdot (z\cdot c_m)\cdot b_k&=e_0\cdot x_1\cdot z \cdot y_m\leq e_0\cdot e\leq e\ . 
	\end{align*}
\item[C4] if $\ell>\alpha(n)=(n+1)\beta(n)-1$ then $b_1\cdot (z\cdot c_m) \cdot b_k\leq e^\sharp$
	(with $\beta$ taken from Lemma~\ref{lemma:conjugacy}).\\
	Since $i_0=0$ and $i_{m+1}=\ell+1$, we have
		$$\ell=m+\sum_{h=0}^m(i_{h+1}-i_h-1).$$
	Since $\ell>(n+1)\beta(n)-1$ this means that either $m>\beta(n)$,
	or $i_{h+1}-i_h-1>n$ for some $h=0\dots m$.
	This means that either $m>\beta(n)$, or $e_h=e^\sharp$ for some $h=0\dots m$.
	In all cases, 
	\begin{align*}
	b_1\cdot (z\cdot c_m)\cdot b_k&=e_0\cdot x_1\cdot z \cdot y_m\leq e^\sharp\ . 
	\end{align*}
	The last inequality can have three origins. Either $m>\beta(n)$ or $e_h=e^\sharp$
	(recall ($\dagger$) stating that $x_h\cdot y_h\leq e\cdot e_h$)
	for some $h=1\dots m$, or $e_0=e^\sharp$. In the two first cases, by Lemma~\ref{lemma:conjugacy},
	$x_1\cdot 	z\cdot y_m\leq e^\sharp$, and thus $e_0\cdot x_1\cdot z\cdot y_m\leq e^\sharp$
	(since $e_0$ is either $1$, or $e$, or $e^\sharp$).
	In the third case,
	$e_0\cdot x_1\cdot z\cdot  y_m\leq e^\sharp$ since $x_1\cdot z\cdot y_m\leq e$.
\end{asparadesc}

Gathering the claims C1, C2, C3,
we get that $u_1,\dots,u_k$ $n$-evaluate to $b_1,\dots,b_k$,
that $b_2\dots b_{k-1}$ $n$-evaluates to $c=z\cdot c_m$ and that $b_1\cdot c\cdot b_k\leq e$.
This is exactly the induction hypothesis for the idempotent node case.
If we further gather C4, we get that if the root node of $T$ is a stabilisation node, 
$b_1\cdot c\cdot b_k\leq e$. Once more the induction hypothesis is satisfied.
\end{proof}


\section{Recognisable cost functions}
\label{section:recognisability}

We have seen in the previous sections the notion of stabilisation monoids,
as well as the key technical tools for dealing with them, namely computations,
over-computations and under-computations.
In particular, we have seen Theorem~\ref{theorem:exists-computation}
and Theorem~\ref{theorem:unicity-computations} that state the existence
of computations and the ``unicity'' of their values.
In this section, we put these notions in action, and introduce the definition
of recognisable cost functions. We will see in particular that 
the hypothesis of Fact~\ref{fact:schema-cost-monadic}
is fulfilled by recognisable cost functions, and as a consequence
the domination problem for cost-monadic logic is decidable over finite
words.

\subsection{Recognisable cost functions}
\label{subsection:definition-recognisability}

Let us fix a stabilisation monoid $\monoid$.
An \intro{ideal} of~$\monoid$ is a subset~$I$ of~$M$ which is downward closed,
{\it i.e.}, such that whenever $a\leq b$ and $b\in I$ we have~$a\in I$.
Given a subset~$X\subseteq M$, we denote by~$X\ideal$
the least ideal which contains $X$, {\it i.e.}, $X\ideal=\{y~:~y\leq x,~x\in X\}$.

The three ingredients used for recognising a cost functions are
a monoid $\monoid$, a mapping~$h$ from letters of the alphabet $\alphabet$ to $M$
(that we extend into a morphism $\tilde h$ from $\alphabet^*$ to $M^*$),
and an ideal~$I$.

We then define for each non-negative integer~$p$ four functions from $\alphabet^*$ to $\NI$:
\begin{center}
\begin{tabular}{ll}
$\sem{\monoid,h,I}^{--}_p(u)$&= $\inf\{n~:~$there exists an $n$-under-computation of value in~$M\setminus I$\\
		&\qquad\qquad of height at most~$p$ for~$\tilde h(u)$ $\}$\\
$\sem{\monoid,h,I}^{-}_p(u)$&= $\inf\{n~:~$there exists an $n$-computation of value in~$M\setminus I$\\
		&\qquad\qquad of height at most~$p$ for~$\tilde h(u)$ $\}$\\
$\sem{\monoid,h,I}^{+}_p(u)$&= $\sup\{n+1~:~$there exists an $n$-computation of value in~$I$\\
		&\qquad\qquad of height at most~$p$ for~$\tilde h(u)$ $\}$\\
$\sem{\monoid,h,I}^{++}_p(u)$&= $\sup\{n+1~:~$there exists an $n$-over-computation of value in~$I$\\
		&\qquad\qquad of height at most~$p$ for~$\tilde h(u)$ $\}$
\end{tabular}
\end{center}
These four functions (for each $p$) are candidates to be recognised by $\monoid,h,I$.
The following lemma shows that if $p$ is sufficiently large, all the four functions
belong to the same cost function.
\begin{lemma}
For all~$p\geq 3|M|$, 
\begin{align*}
\sem{\monoid,h,I}^{--}_p\leq \sem{\monoid,h,I}^{-}_p\leq\sem{\monoid,h,I}^{+}_p\leq\sem{\monoid,h,I}^{++}_p\ .
\end{align*}
For all~$p$, there exists a polynomial~$\alpha$ such that
\begin{align*}
	\sem{\monoid,h,I}^{++}_p\leq\alpha\circ\sem{\monoid,h,I}^{--}_p\ .
\end{align*}
For all~$p\leq r$,
\begin{align*}
\sem{\monoid,h,I}^{--}_r&\leq\sem{\monoid,h,I}^{--}_p\ ,&\text{and}\quad\sem{\monoid,h,I}^{++}_p&\leq\sem{\monoid,h,I}^{++}_r\ .
\end{align*}
\end{lemma}
\begin{proof}
All the inequalities are direct consequences of the results of the previous section.

The middle inequality in the first equation simply comes from 
the fact that, thanks to Theorem~\ref{theorem:exists-computation},
the union of the set of integers over which range the infimum in the definition of $\sem{\monoid,h,I}^{+}_p$
with the set of integers over which range the supremum in the definition of $\sem{\monoid,h,I}^{-}_p$ 
is $\nats$. This suffices for justifying the middle inequality.
The first and last inequalities of the equation simply come from the fact
that each $n$-computation is in particular
an $n$-under-computation and an $n$-over-computation respectively.

The second line is directly deduced from Theorem~\ref{theorem:unicity-computations}.

The third statement  simply comes from the the fact that each $n$-under-computation
of height at most $p$ is also an $n$-under-computation of height $r$ (left inequality).
The same goes for over-computations.
\end{proof}
The main consequence of this lemma takes the form of a definition.
\begin{definition}\label{definition:recognisable}
For all stabilisation monoids $\monoid$, all mappings from an alphabet $\alphabet$
to~$M$, and all ideals $I$, there exists a unique cost function $\sem{\monoid,h,I}$
from $\alphabet^*$ to $\NI$ such that for all $p\geq 3|M|$, all the functions 
$\sem{\monoid,h,I}^{--}_p$,
$\sem{\monoid,h,I}^{-}_p$,
$\sem{\monoid,h,I}^{+}_p$ and
$\sem{\monoid,h,I}^{++}_p$ belong to $\sem{\monoid,h,I}$.
It is called the \intro{cost function recognised by $\monoid,h,I$}.
One sometimes also writes that it  is recognised by $\monoid$.
\end{definition}

\begin{example}\label{example:recognisable-size}
The cost function $|\cdot|_a$ is recognised 
by~$\monoid,\{a\mapsto a,b\mapsto b\},\set{0}$,
where $\monoid$ is the stabilisation monoid of Example~\ref{example:stabilisation-monoid-counta}.
In particular, the informal reasoning developed in the example such as ``a few+a few=a few''
now has a formal meaning: the imprecision in such arguments is absorbed in the
equivalence up to $\alpha$ of computation trees, and results in the fact that the monoid
does not define a unique function, but instead directly a cost function.
\end{example}

Another example is the case of standard regular languages.
\begin{example}\label{example:language-recognisable}
Let us recall that a monoid $\monoid$ together with $h$ from $\alphabet$ to $M$
and a subset $F\subseteq M$ is said to \intro{recognise} a language $L$ over $\alphabet$
if for all words $u$, $u\in L$ if and only if $\pi(\tilde h(u))\in F$.
The same monoid can be seen, thanks to Remark~\ref{remark:standard-monoid}
as a stabilisation monoid. In this case, thanks to Remark~\ref{remark:computations-standard},
the same $\monoid,h,F$ recognises the characteristic mapping of $L$.
\end{example}

An elementary result is also the closure under composition with a morphism.
\begin{myfact}\label{fact:recognisable-composition-morphism}
Let $\monoid,h,I$ recognise a cost function $f$ over $\alphabet^*$,
and let $z$ be a mapping from another alphabet $\alphabetB$ to $\alphabet$, then $\monoid,h\circ z,I$
recognises $f\circ \tilde z$.
\end{myfact}
\begin{proof}
One easily checks that the computations involved in the definition of $\sem{\monoid,h,I}^+(\tilde z(u))$
are exactly the same as the one involved in the definition of  $\sem{\monoid,h\circ z,I}^+(u)$.
\end{proof}

We continue this section by developing other tools for analysing the recognisable cost functions.

\subsection{The $\sharp$-expressions}
\label{subsection:sharp-expression}

We now present the notion of $\sharp$-expressions. This provides a convenient notation
in several situations. This object was introduced by
Hashiguchi for studying distance automata \cite{Hashiguchi82b}.
The $\sharp$-expressions can be seen in two different ways.
On one side, a $\sharp$-expression allows to denote an element in a stabilisation monoid.
On the other side, a $\sharp$-expression denotes an infinite sequence of words.
Such sequences are used as witnesses, {\it e.g.}, of the non-existence of a bound
for a function (if the function tends toward infinity over this sequence),
or of the non-divergence of a function $f$ (if the function is bounded over the sequence).
More generally, $\sharp$-expressions will be used as witnesses 
of non-domination.

In a finite monoid~$\monoid$, given an element $a\in M$,
one denotes by~$a^\omega$ the only idempotent which is a power of $a$.
This element does not exist in general for infinite monoids, while it always does
for finite monoids (our case). Furthermore, when it exists, it is unique.
In particular in a finite monoid $\monoid$, $a^\omega=a^\Omega$, where
$\Omega$ is some multiple of ${|M|!}$.
This is a useful notion since the operator of stabilisation is only defined
for idempotents. In a stabilisation monoid, let us denote by
$a^\osharp$ the element $(a^\Omega)^\sharp$.
As opposed to $a^\sharp$ which is not defined if $a$ is not idempotent,
$a^\osharp$ is always defined. We consider $\Omega$ as fixed from now.

A \intro{$\sharp$-expression} over a set~$A$ is
an expression composed of letters from $A$, products, and exponents with $\osharp$.
A $\sharp$-expression~$E$ over a stabilisation monoid denotes a computation
in this stabilisation monoid. It naturally evaluates to an element of $E$,
denoted~$\val(E)$, and called the \intro{value of $E$}.
A $\sharp$-expression is called \intro{strict} if it contains at least
one occurrence of~$\osharp$.

Given a set $A\subseteq M$, call \intro{$\langle A\rangle^\sharp$}
the set of values of expressions over~$A$.
Equivalently, it is the least set which contains $A$ and is closed under product and
stabilisation of idempotents. One also denotes
\intro{$\langle A\rangle^{\sharp{+}}$} the set of values of strict $\sharp$-expressions
over~$A$.

The~\intro{$n$-unfolding} of a $\sharp$-expression over~$A$ is
a word in~$A^+$ defined inductively as follows:
\begin{align*}
\unf(a,n)	&= a\\
\unf(EF,n)	&= \unf(E,n)\unf(F,n)\\
\unf(E^\osharp,n)	&=\overbrace{\unf(E,n)\dots\unf(E,n)}^{n~\text{times}}
\end{align*}

We conclude this section by showing how
$\sharp$-expressions can be used as witnesses of the behaviour
of a regular cost function.
\begin{proposition}\label{proposition:sharp-witness}
Assume~$\monoid,h,I$ recognises $f$, and
let~$E$ be a $\sharp$-expression over~$\alphabet$ of value~$a$, then:
\begin{iteMize}{$\bullet$}
\item if~$a\in I$ then $\{f(\unf(E,\Omega n))~:~n\geq 1\}$ tends toward infinity,
\item if~$a\not\in I$ then $\{f(\unf(E,\Omega n))~:~n\geq 1\}$ is bounded.
\end{iteMize}
\end{proposition}
\proof
We need, given a positive integer $n$, to produce an $n$-computation of value $a$.
It is defined as follows:
\begin{align*}
\computation(a,n)	  &= a\\
\computation(EF,n) &= \val(EF)[\computation(E,n),\computation(F,n)]\\
\computation(E^\osharp,n)&= \val(E^\osharp)[\overbrace{\computation(E^\Omega,n),\dots,\computation(E^\Omega,n)}^{n~\text{times}}]\\
\computation(E^1,n)&= \computation(E,n)\\
\computation(E^k,n)&= \val(E^k)[\computation(E,n),\computation(E^{k-1},n)]\tag{for $k>1$}
\end{align*}
It is easy to check that $\computation(E,n)$ is an~$n$-computation for~$\unf(u,\Omega n)$, that
its value is $\val(E)$ and that its height $H$ depends only upon $E$.
What we have is in fact stronger:
$\computation(E,n)$ is an $m$-computation for~$\unf(E,\Omega n)$ for all~$m\leq n$.

\begin{asparaitem}
\item
Let us suppose~$a\in I$.
We have:
\begin{align*}
\sem{\monoid,h,I}^{+}_{H}&=\sup\{m+1~:~\text{there is an $m$-computation for~$\tilde h(\unf(E,\Omega n))$ of value in~$I$}\}\\
	&\geq n\ .\tag{\text{using~$\computation(E,n)$ as a witness}}
\end{align*}
Thus~$f(\unf(E,\Omega n))$ tends toward infinity when $n$ tends to infinity.

\item 
Suppose now that if~$a\not\in  I$, one obtains:
\begin{align*}
\sem{\monoid,h,I}^{-}_{H}&=\inf\{m~:~\text{there is an $m$-computation
		for~$\tilde h(\unf(E,\Omega n))$ of value in~$M\setminus I$}\}\\
	&=0\ .\tag{\text{using~$\computation(E,n)$ as a witness}}
\end{align*}
Thus $\{f(\unf(E,n))~:~n\in\nats\}$ is bounded.\qed
\end{asparaitem}

\subsection{Quotients, sub-stabilisation monoids and products}
\label{subsection:morphisms}

We introduce here some useful definitions for manipulating and composing stabilisation
monoids. We use them for proving the closure of recognisable cost functions under
min and max (Corollary~\ref{corollary:rec-min-max}).

Let $\monoid=\langle M,\cdot,\sharp,\leq\rangle$
and $\monoid'=\langle M',\cdot',\sharp',\leq'\rangle$ be stabilisation monoids.
A \intro{morphism of stabilisation monoids}
from  $\monoid$ to $\monoid'$
is a mapping~$\mu$ from~$M$ to~$M'$ such that
\begin{iteMize}{$\bullet$}
\item $\mu(1_\monoid)=1_{\monoid'}$,
\item $ \mu(x)\cdot' \mu(y) = \mu(x\cdot y)$ for all~$x,y$ in~$M$,
\item for all~$x,y$ in~$M$, if $x\leq y$ then $\mu(x)\leq' \mu(y)$,
\item $\mu(e)^{\sharp'} = \mu(e^\sharp)$ for all~$e\in E(\monoid)$
  (in this case, $\mu(e)\in E(\monoid')$).
\end{iteMize}

\begin{remark}\label{remark:morphism}
The $n$-computations ({\it resp.}, $n$-under-computations, $n$-over-computations)
over~$\monoid$ are transformed by morphism (applied to each node of the tree)
into $n$-computations ({\it resp.}, $n$-under-computations, $n$-over-computations)
over~$\monoid'$.
In  a similar way, the image under morphism of a $\sharp$-expression
over~$\monoid$ is a~$\sharp$-expression over~$\monoid'$.
\end{remark}
We immediately obtain:
\begin{lemma}\label{lemma:morphism}
For $\mu$ a morphism of stabilisation monoids from
$\monoid$ to~$\monoid'$, $h$ a mapping 
from an alphabet~$\alphabet$ to~$\monoid$ and~$I'$ an ideal
of $\monoid'$, we have:
\begin{align*}
\sem{\monoid,h,\mu^{-1}(I')}&=\sem{\monoid',\mu\circ h,I'}\ .
\end{align*}
\end{lemma}
\begin{proof} Let us remark first that~$I=\mu^{-1}(I')$
is an ideal of~$\monoid$.

Let~$u$ be a word in~$\alphabet^*$.
Let us consider an $n$-computation over~$\monoid$ for~$\tilde h(u)$
of value~$a\in \mu^{-1}(I')$. This computation can be transformed
by morphism into an $n$-computation over~$\monoid'$
for $(\mu\circ h)(u)$  of value~$\mu(a)\in I'$.
In a similar way, each $n$-computation over~$\tilde h(u)$
of value~$a\in M\setminus \mu^{-1}(I')$
can be transformed into an $n$-computation of value
$\mu(a)\in M'\setminus I'$.
\end{proof}

The notion of morphism is intimately related to the notion of product.
Given two stabilisation monoids $\monoid=\langle M,\cdot,\sharp,\leq\rangle$
and~$\monoid'=\langle M',\cdot',\sharp',\leq'\rangle$,
one defines their \intro{product} by:
\begin{align*}
\monoid\times \monoid'&=\langle M\times M',\cdot'',\sharp'',\leq ''\rangle
\end{align*}
where $(x,x')\cdot''(y,y')=(x\cdot y,x'\cdot' y')$, $(e,e')^{\sharp''}=(e^\sharp,{e'}^{\sharp'})$,
and~$(x,x')\leq''(y,y')$ if and only if $x\leq y$ and~$x'\leq' y'$.

As expected, the projection over the first component
({\it resp.}, second component) is a morphism  of stabilisation monoids
from $\monoid\times\monoid'$ onto $\monoid$ ({\it resp.}, onto~$\monoid'$).
It follows by Lemma~\ref{lemma:morphism}
that if $f$ is recognised by~$\monoid,h,I$ and~$g$ by
$\monoid',h',I'$, then $f$ is also recognised
by~$\monoid\times\monoid',h\times h',I\times M'$
and $g$ by~$\monoid\times\monoid',h\times h',M\times I'$, 
in which one sets 
$(h\times h')(a) = (h(a),h(a'))$ for all letters~$a$.
Thus one obtains:
\begin{lemma}\label{lemma:joint-recognition}
If~$f$ and~$g$ are recognisable cost functions
over $\alphabet^*$,
there exist a stabilisation monoid~$\monoid$,
an application $h$ from~$\alphabet$ to~$\monoid$
and two ideals~$I,J$ such that~$\monoid,h,I$ recognises~$f$
and $\monoid,h,J$ recognises~$g$.
\end{lemma}
\begin{corollary}\label{corollary:rec-min-max}
If~$f$ and~$g$ are recognisable, then so are $\max(f,g)$
and $\min(f,g)$.
\end{corollary}
\begin{proof}
According to Lemma~\ref{lemma:joint-recognition},
one assumes $f$ recognised by $\monoid,h,I$ and~$g$ by~$\monoid,h,J$.
Then (for a height fixed to at most $p=3|M|$) one has:
\begin{align*}
\max&(\sem{\monoid,h,I}^{+}_p,\sem{\monoid,h,J}^{+}_p)(u)\\
&=\max\begin{cases}
	\sup\{n+1~:~\text{there is an $n$-computation for~$\tilde h(u)$ of value in~$I$}\}&\\
	\sup\{n+1~:~\text{there is an $n$-computation for~$\tilde h(u)$ of value in~$J$}\}
	\end{cases}\\
&=\sup\{n+1~:~\text{there exists an $n$-computation over~$\tilde h(u)$ of value in~$I\cup J$}\}\\
&=\sem{\monoid,h,I\cup J}^{+}_p\ .
\end{align*}
Thus  $\max(f,g)$ is recognised by~$\monoid,h,I\cup J$.
In a similar way, $\min(f,g)$ is recognised by~$\monoid,h,I\cap J$.
\end{proof}

\subsection{Decidability of the domination relation}
\label{subsection:domination-recognizable}

We are now ready to establish the
decidability of the domination relation.
\begin{theorem}\label{theorem:rec-domination}
The domination relation ($\preccurlyeq$) is decidable over
recognisable cost functions.
\end{theorem}
\begin{proof}
Let $f,g$ be recognisable cost functions.
According to Lemma~\ref{lemma:joint-recognition} 
there exist a stabilisation monoid~$\monoid$, a mapping~$h$
from the alphabet~$\alphabet$ to~$M$
and two ideals~$I,J$ such that~$\monoid,h,I$ recognises~$f$ and~$\monoid,h,J$
recognises~$g$.

We show that $f$ dominates $g$ if and only if the following (decidable) property holds:
\begin{align*}
\langle h(\alphabet)\rangle^\sharp\cap I\subseteq J\ .
\end{align*}

\noindent
{\em First direction.} Let us suppose
$\langle h(\alphabet)\rangle^\sharp\cap I\subseteq J$.
Of course, every $n$-computation over $\tilde h(u)$ for a word~$u$ over $h(\alphabet)$
has its value in $\langle h(\alphabet)\rangle^\sharp$.
It follows that (for heights at most $3|M|$):
\begin{align*}
\sem{\monoid,h,I}^{+}(u)&= \sup\{n+1~:~\text{there is an~$n$-computation for~$h(u)$ of value in~$I$}\}\\
 &\leq \sup\{n+1~:~\text{there is an~$n$-computation for~$h(u)$ of value in~$J$}\}\\
 &=\sem{\monoid,h,J}^{+}(u)\ .
\end{align*}
Thus we have $f\preccurlyeq g$.
\smallskip

\noindent
{\em Second direction.} Let us suppose the existence of
$a\in \langle h(\alphabet)\rangle^\sharp\cap I\setminus J$.
By definition of~$\langle h(\alphabet)\rangle^\sharp$,
there is a $\sharp$-expression $E$ over~$h(\alphabet)$ of value~$a$.
Let~$F$ be the $\sharp$-expression over~$\alphabet$ obtained by
substituting to each element~$x\in h(\alphabet)$ some
letter from~$c\in\alphabet$ such that~$h(c)=x$.
According to Proposition~\ref{proposition:sharp-witness}, $f$ is 
unbounded over~$\{\unf(F,\Omega n)~:~n\geq 1\}$ (for some suitable $k$).
However, still applying Proposition~\ref{proposition:sharp-witness},
$g$ is bounded over~$\{\unf(F,\Omega n)~:~n\geq 3\}$.
This witnesses that $g$ does not dominate~$f$.
\end{proof}

\subsection{Closure under $\inf$-projection}
\label{subsection:closure-inf-projection}

We establish the following theorem.
\begin{theorem}\label{theorem:inf-projection}
Recognisable cost functions are effectively closed under inf-projection.
\end{theorem}
The projection in the classical case (of recognisable languages of finite words) requires
a powerset construction (for monoids as for deterministic automata).
In our case, the approach is similar.
Let $z$ be a mapping from alphabet $\alphabet$ to $\alphabetB$.
The goal of a stabilisation monoid which would recognise the inf-projection by $z$ 
of a recognisable cost function is to keep track of all the values a computation
could have taken for some inverse image by $\tilde z$ of the input word.
Hence an element of the stabilisation monoid for the inf-projection of the
cost function consists naturally of a set of elements in the original monoid.

A closer inspection reveals that it is possible to close these subsets downward,
{\it i.e.}, to consider only ideals. In fact, it is not only possible, but it is even necessary
for the construction to go through. Let us describe more formally this construction.

We have to consider a construction of ideals.
Let~$M_\ideal$ be the set of ideals of $\monoid$.
One equips~$M_\ideal$ of an order simply by inclusion:
$$
I\leq J\qquad\text{if}\qquad I\subseteq J\ ,
$$
and of a product as follows:
\begin{align*}
A\cdot B\quad&=\quad\{a\cdot b~:~a\in A,~b\in B\}\ideal\ .
\end{align*}
Finally, the stabilisation is defined for an idempotent by:
\begin{align*}
E^\sharp&=\langle E \rangle^{\sharp+}\ideal\ .
\end{align*}
The resulting structure $\langle M_\ideal,\cdot,\leq,\sharp\rangle$
is denoted $\monoid_\ideal$.

It may seem a priori that our first goal would be to prove that
the structure defined in the above way is indeed a stabilisation monoid.
In fact, thanks to Proposition~\ref{proposition:unicity-to-axioms}, this
will be for free (see Lemma~\ref{lemma:ideal-is-stabilisation-monoids} below).

We now prove that $\monoid_\ideal$
can be used for recognising the inf-projection of a cost
function recognised by $\monoid$. Thus our goal is to relate
the (under)-computations in $\monoid_\ideal$ to the (under)-computations
in $\monoid$. This will provide a semantic link between the two stabilisation monoids.
This relationship takes the form of Lemmas~\ref{lemma:inf-key1}
and~\ref{lemma:inf-key2} below.

Let us first state a simple remark on the structure of idempotents.
\begin{lemma}\label{lemma:inf-idempotent}
If~$E$ is an idempotent in $M_\downarrow$, then for all~$a\in E$
there exist $b,c,e\in E$ with $e$ idempotent such that~$a\leq b\cdot e\cdot c$.
\end{lemma}
\begin{proof}
As~$E=\overbrace{E\cdots E}^{n~\text{times}}$, for all~$n\geq 1$, there exist~$a_1,\dots,a_n\in E$
such that $a\leq a_1\cdots a_n$.
Hence using Ramsey's theorem, for~$n$ sufficiently
large, there exist~$1<i\leq j< n$ such that~$a_i\cdots a_{j}=e$ is an idempotent.
One sets $b=a_1\cdots a_{i-1}$, $c=a_{j+1}\cdots a_n$. We have $b,c,e\in E$,
and $a\leq b\cdot e\cdot c$.
\end{proof}
A similar characterisation holds for the stabilisation of idempotents.
\begin{lemma}\label{lemma:inf-idempotent-stable}
If~$E$ is a stable idempotent in $M_\downarrow$ ({\it i.e.} such that $E=E^\sharp$), then for all~$a\in E$
there exist $b,c,e\in E$ with $e$ idempotent such that~$a\leq b\cdot e^\sharp\cdot c$.
\end{lemma}
\begin{proof}
By definition of $E^\sharp$, there is a strict $\sharp$-expression $F$ over $E$ such that $a\leq \val(F)$.
Thus it is sufficient to prove, by induction, that all strict $\sharp$-expressions $F$
is such that $\val(F)\leq b\cdot e^\osharp\cdot c$ for some $b,e,c$ in $E$.
The base case is $F=G^\osharp$ where $G$ is a non strict $\sharp$-expression.
In this case $\val(G)=g\in E$.
It follows that $\val(G^\osharp)=\val(G)^\osharp\leq g^\omega\cdot g^\osharp\cdot g^\omega$.
Thus the induction hypothesis holds. The other case is the product $F=GH$.
By induction hypothesis, $\val(G)\leq a\cdot e^\osharp\cdot b$ and
$\val(H)\leq a'\cdot f^\osharp\cdot b'$. It follows that  $\val(F)\leq a\cdot e^\osharp\cdot d$
with $d=a'\cdot f ^\osharp\cdot b'\in E$. Once more the induction hypothesis hold.
\end{proof}

\begin{lemma}\label{lemma:inf-key1}
Let~$A_1\dots A_k$ be a word over~$M_\ideal$ and let $T$ be an $n$-under-computation
over $A_1\dots A_k$ of height at most~$p$ and of value~$A$.
For all~$a\in A$, there exists an $n$-under-computation of height~$3p$ and value~$a$
for some word~$a_1\dots a_k$ such that $a_1\in A_1$,\dots,$a_k\in A_k$.
\end{lemma}
\begin{proof}
The proof is by induction on~$p$.
\medskip

\noindent
\emph{Leaf case}, {\it i.e.}, $T=A_1$. Let~$a\in A\subseteq A_1$, then $a$ is an $n$-computation of value~$a\in A$.
\medskip

\noindent
\emph{Binary node}, {\it i.e.},  $T=A[T_1,T_2]$. Let $B_1$ and $B_2$
be the respective values of $T_1$ and $T_2$.
Let~$a\in A\subseteq B_1\cdot B_2$. By definition of the product, there exists~$b_1\in B_1$
and~$b_2\in B_2$ such that~$a\leq b_1\cdot b_2$.
By induction hypothesis, there exist $n$-under-computations~$t_1$ and~$t_2$
of respective values~$b_1$ and~$b_2$.
The $n$-under-computation $a[t_1,t_2]$ satisfies the induction hypothesis.
\medskip

\noindent
\emph{Idempotent node.}
$T=F[T_1,\dots,T_k]$ for some~$k\leq n$ where $F\subseteq E$
	for an idempotent $E$ such that the value of $T_i$ is $E$ for all~$i$.
	Let~$a\in F\subseteq E$.
	We have $a\leq b\cdot e\cdot c$ for some $b,c,e\in E$ (Lemma~\ref{lemma:inf-idempotent}).
	We then apply the induction hypothesis for~$b,e,\dots,e$ and~$c$
	on the $n$-under-computations~$T_1,\dots,T_{k-1}$ and $T_k$ respectively,
	yielding the $n$-under-computations $t_1,\dots,t_{k-1}$ and $t_k$ respectively.
	The tree $a[t_1,(e\cdot c)[e[t_2,\dots,t_{k-1}],t_k]]$ is an $n$-under-computation
	witnessing that the induction hypothesis holds.
\medskip

\noindent
\emph{Stabilisation node.}
$T=F[T_1,\dots,T_k]$ for some~$k> n$ and $F\subseteq E^\sharp$ for some idempotent $E$
	such that the value of $T_i$ is $E$ for all~$i$.
	Let~$a\in F\subseteq E^\sharp$.
	We have $a\leq b\cdot e^\sharp\cdot c$ for some $b,e,c\in E$ (Lemma~\ref{lemma:inf-idempotent-stable}).
	We then apply the induction hypothesis for~$b,e,\dots,e$ and~$c$ respectively
	and the computations~$T_1,\dots,T_k$ respectively,
	yielding the $n$-under-computations $t_1,\dots,t_k$ respectively.
	We conclude by constructing the 
	$n$-under-computation~$a[t_1,(e^\sharp\cdot c)[e^\sharp[t_2,\dots,t_{k-1}],t_k]]$
	(remark that $e^\sharp[t_2,\dots,t_{k-1}]$ is a valid under-computations since
	$e^\sharp\leq e$).
\end{proof}

\begin{lemma}\label{lemma:inf-key2}
There exists a polynomial $\alpha$ such that 
for all words~$A_1\dots A_k$ over~$M_\ideal$ and all $\alpha(n)$-over-computation $T$ 
for~$A_1\dots A_k$ of height~$p$ and value~$A$,
and all $a_1\in A_1,\dots,a_k\in A_k$,
there exists an~$n$-computation over~$a_1\dots a_k$ of value~$a\in A$ and of height at most~$3|M|p$.
\end{lemma}
\begin{proof}
The proof is by induction on $p$. Set $\alpha(n)=n^{3|M|}$ for all $n$.
\medskip

\noindent
\emph{Leaf case,} {\it i.e.}, $T=A$ and~$u=a_1\in A_1\subseteq A$. Hence~$a_1$ is a computation
satisfying the induction hypothesis.
\medskip

\noindent
\emph{Binary node,} {\it i.e.}, $T=A[T_1,T_2]$ where $T_1$ and $T_2$ have respective
	values $B_1$ and $B_2$ such that $B_1\cdot B_2\subseteq A$. One applies the induction
	hypothesis on $T_1$ and $T_2$, and gets computations $t_1$
	and $t_2$, of respective values $b_1\in B_1$ and $b_2\in B_2$.
	The induction hypothesis is then fulfilled with the $n$-computation
	 $(b_1\cdot b_2)[t_1,t_2]$ of value~$b_1\cdot b_2\in A$.
\medskip

\noindent
\emph{Idempotent node,} {\it i.e.},
	$T=F[T_1,\dots,T_k]$ for~$k\leq n^{3|M|}$ where $T_1,\dots,T_k$
	share the same idempotent value $E\subseteq F$.
	Let~$t_1,\dots,t_k$ be the $n$-computations of respective values $b_1,\dots,b_k$
	obtained by applying the induction hypothesis
	on $T_1,\dots,T_k$ respectively. Furthermore, according to Theorem~\ref{theorem:exists-computation},
	there exists an $n$-computation $t$ for the word $b_1\dots b_k$ of height at most $3|M|$.
	Let $a$ be the value of $t$. Since $E$ is an idempotent, it is closed under product and stabilisation
	and contains $b_1,\dots,b_k$. It follows that $a\in E$ (by induction on the height of $t$).
	The induction hypothesis holds using the witness $n$-computation $t\{t_1,\dots,t_k\}$
	where $t\{t_1,\dots,t_k\}$ is obtained from $t$ by substituting the $i$th leaf for $t_i$ for all~$i=1\dots k$.
\medskip

\noindent
\emph{Stabilisation node,} {\it i.e.},
	$T=F[T_1,\dots,T_k]$ for~$k> n^{3|M|}$ where $T_1,\dots,T_k$
	all share the same idempotent value $E$ such that $E^\sharp\subseteq F$.
	Let~$t_1,\dots,t_k$ be the $n$-computations of respective values $b_1,\dots,b_k$ obtained by applying
	the induction hypothesis on $T_1,\dots,T_k$ respectively.
	Furthermore, according to Theorem~\ref{theorem:exists-computation},
	there exists an $n$-computation $t$ for $b_1\dots b_k$ of height at most $3|M|$.
	Since $t$ has height at most $3|M|$ and has more than $n^{3|M|}$ leaves,
	it contains at least one node of degree more than $n$, namely, a  stabilisation node.
	It is then easy to prove by induction on the height of $t$ that the value of $t$
	belongs to $\langle E\rangle^{\sharp+}$.
	Thus the $n$-computation $t\{t_1,\dots,t_k\}$ (as defined in the above case) is a witness
	for the induction hypothesis.
\end{proof}

\begin{lemma}\label{lemma:ideal-is-stabilisation-monoids}
$\monoid_\ideal$ is a stabilisation monoid.
\end{lemma}
\begin{proof}
Consider an $n$-under-computation $T$ of value $A$ in $\monoid_\ideal$
for some word $A_1\dots A_k$ of height at most $p$,
and some $\alpha(n)$-over-computation $T'$ for the
same word of value $B$ (with $\alpha(n)=\alpha'(n)^{3|M|}$ where $\alpha'$
	is obtained from Theorem~\ref{theorem:unicity-computations} applied to $\monoid$ for height at most $3p$).
We aim at $A\leq B$. Indeed, this implies that one can use Proposition~\ref{proposition:unicity-to-axioms},
and get that $\monoid_\ideal$ is a stabilisation semigroup (it is then straightforward to
prove it a stabilisation monoid. 

Let $a\in A$, we aim at $a\in B$, thus proving $A\subseteq B$, {\it i.e.}, $A\leq B$.
By Lemma~\ref{lemma:inf-key1}, there exists
an $n$-under-computation for some word $a_1\dots a_k$ with $a_1\in A_1,\dots,a_k\in A_k$,
of height at most $3p$ and value $a$.
Applying Lemma~\ref{lemma:inf-key2} on $T'$, there exists an $\alpha'(n)$-over-computation
for the same word $a_1\dots a_k$ of value $b\in B$. Then applying Theorem~\ref{theorem:unicity-computations},
we get $a\leq b$. Hence $a\in B$ since $B$ is downward closed.
\end{proof}

We can now establish the result of this section.
\begin{proof}[Proof of Theorem~\ref{theorem:inf-projection}]
Assume a function $f$ over the alphabet $\alphabet$
is recognised by~$\monoid,h,I$,
and let~$z$ be some mapping from~$\alphabet$ to~$\alphabetB$.
Let $H$ be the mapping from $\alphabetB$ to $M_\ideal$ 
which to $b\in\alphabetB$ associates $h(z^{-1}(b))\ideal$,
and let~$K\subseteq M_\ideal$ be
$$
K=\{J\in M_\ideal~:~J\subseteq I\}\ .
$$
We shall prove that $M_\ideal,H,K$ recognise the cost function of $f_{\inf,z}$.
There are two directions.

Consider a word $u=b_1\dots b_k$ over alphabet $\alphabetB$ such that
$$\sem{\monoid_\ideal,H,K}^{--}_{3|M_\ideal|}(u)\leq n\ .$$
This means that there exists an $\alpha(n)$-under-computation over $A_1\dots A_k=\tilde H(u)$
of value $J\in M_\ideal\setminus K$, and height at most $3|M_\ideal|$. 
By definition of $K$, this means that there exists some $a\in J\setminus I$.
Let us apply now Lemma~\ref{lemma:inf-key1}, and obtain
an $n$-under-computation for some $a_1\dots a_k$ of value $a\in A$,
and of height at most $d=9|M_\ideal|$, where $a_i\in A_i$ for all $i=1\dots k$.
By definition of $H$, there exists $c_i\in\alphabet$ such that $z(c_i)=b_i$ and $h(c_i)=a_i$.
Thus, set $v=c_1\dots c_k$. The word $v$ is such that $\tilde z(v)=u$ and $\sem{\monoid,h,I}_d^{--}(v)\leq n$.
This a witness that 
$$
(\sem{\monoid,h,I}_d^{--})_{\inf,z}(u)\leq n\ .
$$

Conversely, consider a word $u$ over the alphabet $\alphabetB$
such that
$$\sem{\monoid_\ideal,H,K}^{++}_{3|M_\ideal|}(u)\geq \alpha(n)+1\ ,$$
where $\alpha$ is the polynomial obtained from Lemma~\ref{lemma:inf-key2}.
This means that there exists an $\alpha(n)$-over-computation for the word $A_1\dots A_k=\tilde H(u)$ of height
at most $3|M_\ideal|$ and value $J\in K$. 
Consider now some word $v=c_1\dots c_k$ over the alphabet $\alphabet$
such that $\tilde z(v)=u$. Let $a_1\dots a_k=\tilde h(v)$.
according to Lemma~\ref{lemma:inf-key2} there exists an $n$-computation for $a_1\dots a_k$
of height at most $d=9|M||M_\ideal|$ and value $a\in J$. Since by definition of $K$, $J\subseteq I$,
this means that $a\in I$. 
Thus, this computation witnesses that $\sem{\monoid,h,I}^{+}_{d}(v)\geq n$.
Since this this holds for all words $v$ such that $\tilde z(v)=u$, we obtain that
$$
(\sem{\monoid,h,I}^{+}_{d})_{\inf,z}(u)\geq n\ .
$$
\end{proof}

\subsection{Closure under $\sup$-projection}
\label{subsection:closure-sup-projection}
We now establish the closure under sup-projection, using a proof
very similar to the previous section.

\begin{theorem}\label{theorem:sup-projection}
Recognisable cost functions are effectively closed under sup-projection.
\end{theorem}

The closure under $\sup$-projection follows the same principle as the
closure under $\inf$-projection. It uses also a powerset construction.
However, since everything is reversed, this is a construction of co-ideals.
A \intro{co-ideal} is a subset of a stabilisation monoid which is upward closed.
Let~$\monoid_\coideal$ be set of co-ideals over a given stabilisation
monoid~$\monoid$.
Given a set $A$, let us denote by~$A\coideal$ 
the least co-ideal containing~$A$, {\it i.e.}, $A\coideal=\{y~:~y\geq x\in A\}$.
One equips~$\monoid_\coideal$ of an order by:
$$
I\leq J\qquad\text{if}\quad I\supseteq J\ ,
$$
of a product with:
\begin{align*}
A\cdot B\quad&=\quad\{a\cdot b~:~a\in A,~b\in B\}\coideal\ ,
\end{align*}
and of a stabilisation operation by:
\begin{align*}
E^\sharp&=\langle E \rangle^{\sharp+}\coideal\ .
\end{align*}
Let us call $\monoid_\coideal$ the resulting structure. You can remark that $\langle E \rangle^{\sharp+}\coideal=\langle E \rangle^{\sharp}\coideal$. This was not the case for ideals. 
The proof is extremely close to the case of $\inf$-projection.
However, a careful inspection would show that all computations of bounds, and even 
some local arguments need to be modified.

Let us state a simple remark on the structure of idempotents.
\begin{lemma}\label{lemma:sup-idempotent}
If~$E$ is an idempotent in $M_\uparrow$, then for all~$a\in E$
there exist $b,c,e\in E$ with $e$ idempotent
such that $a\geq b\cdot e\cdot c$.
\end{lemma}
\begin{proof} 
As~$E=E\cdots E$, for all~$n$, there exist~$a_1,\dots,a_n\in E$
such that $a\geq a_1\cdots a_n$. Using Ramsey's theorem, for~$n$ sufficiently
large, there exist~$1<i\leq j< n$ such that~$a_i\cdots a_{j}=e$ is an idempotent.
One sets $b=a_1\cdots a_{i-1}$, $c=a_{j+1}\cdots a_n$. We have $b,c,e\in E$,
and $a\geq b\cdot e\cdot c$.
\end{proof}

Our second preparatory lemma is used for the treatment of stabilisation nodes.
\begin{lemma}\label{lemma:sup-stabilisation}
There exists a polynomial $\alpha$
such that for all idempotents~$E$ of $M_\uparrow$, all~$a\in E^\sharp$
and all $\alpha(m)\leq n$, there exists
an $m$-over-computation of height at most $2|M|+3$ of value $a$ over some word over $E$
of length $n$.
\end{lemma}
\begin{proof} 
It is sufficient to prove the result for a single pair $E,a$, and construct for each such case
a polynomial $\alpha_{E,a}$. Then, since there are finitely many such pairs $(E,a)$, one can choose a polynomial $\alpha$
that is above all the $\alpha_{E,a}$. This $\alpha$ will witness the lemma for all choices of $E$ and $a$.

We first claim $(\star)$ that if $b\in E$ then for all $n\geq 1$, there exists a word $w_n$ over $E$ of length
$n$ such that for all $m\geq 1$ there exist an $m$-over-computation for $w$ of height at most~$3$ of value $b$.
Indeed, by Lemma~\ref{lemma:sup-idempotent},
$b\geq c\cdot f\cdot d$ where $c,f,d$ belong to $E$, and $f$ is idempotent.
So if $n=1$, we take $w=b$. If $n=2$, we take $w=c(f \cdot d)$.
Finally, for $n\geq 3$, there is a natural $m$-over-computation of value $c\cdot f\cdot d$ of height $2$ or $3$ over
the word $c\underbrace{f\dots f}_{n-2~\text{times}}d$, which is of length $n$.

Consider now a $\sharp$-expression $e$ of value $a\in E^\sharp$ for some idempotent $E$.
Without loss of generality, we can choose it of height at most $2|M|$.
and consider the word $u_m=\unf(e,|M|!(m+1))$ for all $m\geq 1$.
The length of this word is a polynomial $\alpha(m)$, and there is an $m$-over-computation of height at most $2|M|$
of value $a$ for this word.
Consider now some~$n\geq\alpha(m)$. The word $u_m$ can be written $vb$ for some $b\in E$. We can apply the above claim to 
$b$ and $n+1-\alpha(m)$, yielding the word $w_{n+1-\alpha(m)}$. Combining the two $m$-over-computations, we then naturally obtain an $m$-over-computation for the word $vw_{n+1-\alpha(m)}$ of height at most $2|M|+3$ and of value $a$, and this word has length $n$.
\end{proof}

\begin{lemma}\label{lemma:sup-key1}
There exists a polynomial $\alpha$ such that for all
words $A_1\dots A_k$ over~$\monoid_\coideal$ and all  $\alpha(n)$-over-computations $T$
for $A_1\dots A_k$ of height at most~$p$ of value~$A$
and all~$a\in A$, there exists an $n$-over computation of height at most $(2|M|+3)p$ and value~$a$
for some word~$a_1\dots a_k$ with $a_1\in A_1$,\dots,$a_k\in A_k$.
\end{lemma}
\begin{proof}
The proof is by induction on~$p$. We take the polynomial $\alpha$ of Lemma~\ref{lemma:sup-stabilisation}.
\medskip

\noindent
\emph{Leaf case}, {\it i.e.}, $T=A_1$. Let~$a\in A\subseteq A_1$, then $a$ is an $n$-computation of value~$a\in A$.
\medskip

\noindent
\emph{Binary node}, {\it i.e.},  $T=A[T_1,T_2]$. Let $B_1$ and $B_2$
be the respective values of $T_1$ and $T_2$.
Let~$a\in A\subseteq B_1\cdot B_2$. By definition of the product, there exists~$b_1\in B_1$
and~$b_2\in B_2$ such that~$a\geq b_1\cdot b_2$.
By induction hypothesis, there exist $n$-computations~$t_1$ and~$t_2$
of respective values~$b_1$ and~$b_2$.
The $n$-over-computation $a[t_1,t_2]$ satisfies the induction hypothesis.
\medskip

\noindent
\emph{Idempotent node.}
$T=F[T_1,\dots,T_k]$ for some~$k\leq\alpha(n)$ where $F\subseteq E$
	for an idempotent $E$ such that the value of $T_i$ is $E$ for all~$i$.
	Let~$a\in F\subseteq E$.
	We have $a\geq b\cdot e\cdot c$ for some $b,c,e\in E$ (Lemma~\ref{lemma:sup-idempotent}).
	We then apply the induction hypothesis for~$b,e,\dots,e$ and~$c$
	on the computations~$T_1,\dots,T_{k-1}$ and $T_k$ respectively,
	yielding the $n$-under-computations $t_1,\dots,t_{k-1}$ and $t_k$ respectively.
	We conclude by constructing the 
	$n$-under-computation~$a[t_1,(e\cdot c)[e[t_2,\dots,t_{k-1}],t_k]]$.
\medskip

\noindent
\emph{Stabilisation node.}
$T=F[T_1,\dots,T_k]$ for some~$k>\alpha(n)$ and $F\subseteq E^\sharp$ for some idempotent $E$
	such that the value of $T_i$ is $E$ for all~$i$.
	Let~$a\in F\subseteq E^\sharp$.
	According to Lemma~\ref{lemma:sup-stabilisation}, there exists a word $a_1\dots a_k$ over $E$
	and an $n$-over-computation~$t$ for $a_1\dots a_k$ of value $a$ and height at most $2|M|+3$.
	We then apply the induction hypothesis for each of $a_1,\dots, a_k$ with the computations~$T_1,\dots,T_k$ respectively.
	This yields $n$-over-computations $t_1,\dots,t_k$ respectively.
	We conclude by constructing the 
	$n$-over-computation obtained by substituting in $t$ the $i$th leaf with $t_i$.
\end{proof}

\begin{lemma}\label{lemma:sup-key2}
Let~$A_1\dots A_k$ be a word over~$\monoid_\coideal$ and $T$ be an $n$-under-computation~$T$
for~$A_1\dots A_k$ of height~$p$ and value~$A$.
For all words~$u=a_1\dots a_k$ with  $a_1\in A_1,\dots,a_k\in A_k$,
there exists an~$n$-computation over~$a_1\dots a_k$ of value~$a\in A$ and of height at most~$3|M|p$.
\end{lemma}
\begin{proof}
The proof is by induction on $p$.
\medskip

\noindent
\emph{Leaf case,} {\it i.e.}, $T=A$ and~$u=a_1\in A_1\subseteq A$. Hence~$a_1$ is a computation
satisfying the induction hypothesis.
\medskip

\noindent
\emph{Binary node,} {\it i.e.}, $T=A[T_1,T_2]$ where $T_1$ and $T_2$ have respective
	values $B_1$ and $B_2$ such that $B_1\cdot B_2\subseteq A$. One applies the induction
	hypothesis on $T_1$ and $T_2$, and get computations $t_1$
	and $t_2$, of respective values $b_1\in B_1$ and $b_2\in B_2$.
	The induction hypothesis is then fulfilled with the $n$-computation
	 $(b_1\cdot b_2)[t_1,t_2]$ of value~$b_1\cdot b_2\in A$.
\medskip

\noindent
\emph{Idempotent node,} {\it i.e.},
	$T=F[T_1,\dots,T_k]$ for~$k\leq n$ where $T_1,\dots,T_k$
	share the same idempotent value $E\subseteq F$.
	Let~$t_1,\dots,t_k$ be the $n$-computations of respective values $b_1,\dots,b_k$
	obtained by applying the induction hypothesis
	on $T_1,\dots,T_k$ respectively. Furthermore, according to Theorem~\ref{theorem:exists-computation},
	there exists an $n$-computation $t$ for the word $b_1\dots b_k$ of height at most $3|M|$.
	Let $a$ be the value of $t$. Since $E$ is an idempotent, it is closed under product
	and contains $b_1,\dots,b_k$. Since furthermore $t$ does not contain any node of stabilisation,
	we obtain that $a\in E$ (by induction on the height of $t$).
	We conclude using the $n$-computation $t\{t_1,\dots,t_k\}$ (obtained from $t$ by substituting
	the $i$th leaf of $t$ for $t_i$) which satisfies the induction hypothesis.
\medskip

\noindent
\emph{Stabilisation node,} {\it i.e.},
	$T=F[T_1,\dots,T_k]$ for~$k> n$ where $T_1,\dots,T_k$
	all share the same idempotent value $E\subseteq F$.
	Let~$t_1,\dots,t_k$ be the $n$-computations of respective values $b_1,\dots,b_k$ obtained by applying
	the induction hypothesis on $T_1,\dots,T_k$ respectively.
	Furthermore, according to Theorem~\ref{theorem:exists-computation},
	there exists an $n$-computation $t$ for $b_1\dots b_k$ of height at most $3|M|$ and value $a$.
	Since $E^\sharp$ is a sub-stabilisation monoid of $\monoid$ which contains $b_1,\dots,b_k$,
	$a$ also belongs to $E^\sharp$.
	Thus the $n$-computation $t\{t_1,\dots,t_k\}$ satisfies the induction hypothesis.
\end{proof}

\begin{lemma}\label{lemma:coideal-is-stabilisation-monoids}
$\monoid_\coideal$ is a stabilisation monoid.
\end{lemma}
\begin{proof}
Consider an $n$-under-computation $T$ of value $A$ in $\monoid_\coideal$
for some word $A_1\dots A_k$ of height at most $p$,
and some $\alpha(n)$-over-computation $T'$ for the
same word of value $B$ (with $\alpha(n)=\alpha'(n)+2$ where $\alpha'$
	is obtained from Theorem~\ref{theorem:unicity-computations} applied to $\monoid$ for height at most $3p$).
We aim at $A\leq B$, which means $A\supseteq B$. Indeed, this implies that one can use Proposition~\ref{proposition:unicity-to-axioms},
and get that $\monoid_\coideal$ is a stabilisation monoid.

Let $b\in B$, we aim at $b\in A$, thus proving $B\subseteq A$, {\it i.e.}, $A\leq B$.
By Lemma~\ref{lemma:sup-key1}, there exists
an $\alpha'(n)$-over-computation for some word $a_1\dots a_k$ with $a_1\in A_1,\dots,a_k\in A_k$,
of height at most $3p$ and value $b$.
Applying Lemma~\ref{lemma:sup-key2} on $T$, there exists an $n$-computation
for the same word $a_1\dots a_k$ of value $a\in A$. Then applying Theorem~\ref{theorem:unicity-computations},
we get $a\leq b$. Hence $b\in A$ since $A$ is upward closed.
\end{proof}

We are ready to complete the proof of closure under sup-projection.
\begin{proof}{Proof of Theorem~\ref{theorem:sup-projection}}
Assume a function $f$ over the alphabet $\alphabet$
is recognised by~$\monoid,h,I$
and let~$z$ be some mapping from~$\alphabet$ to~$\alphabetB$.
Construct $H$ from $\alphabetB$ to $M_\coideal$ that maps
$b\in\alphabetB$ to $h(z^{-1}(b))\coideal$,
and let~$K\subseteq M_\coideal$ be
$$
K=\{J\in M_\coideal~:~J\cap I\neq \emptyset\}\ .
$$
Let us prove that $\monoid_\coideal,H,K$ recognises the cost function of $f_{\sup_z}$.

Consider now a word $u=b_1\dots b_k$ over $\alphabetB$.
Assume 
$$\sem{\monoid_\coideal,H,K}^{++}_{3|M_\coideal|}(u)\geq \alpha(n)+1\ ,$$
where $\alpha$ is the polynomial from Lemma~\ref{lemma:sup-key1}.
This means that there is an $\alpha(n)$-over-computation for~$H(b_1)\dots H(b_k)$ of value~$J\in K$
of height at most $3|M_\coideal|$.
By definition of $K$, there exists $a\in J\cap I$. Hence, by Lemma~\ref{lemma:sup-key1} ,
there exists an $n$-over-computation of height at most $d=(2|M|+3)2|M_\coideal|$ of value $a$
for some $a_1\dots a_k$ with $a_1\in H(b_1),\dots,a_k\in H(b_k)$.
By definition of $H$, this means that there exist $c_i\in\alphabet$ such that $a_i=h(c_i)$ and $z(c_i)=b_i$ for all $i=1\dots k$.
The obtained word $v=c_1\dots c_k$ is such that $\tilde z(v)=u$, and $\sem{\monoid,h,I}^{++}_{d}(v)\geq n$.
This witnesses that
$$
(\sem{\monoid,h,I}^{++}_{d})_{\sup,z}(u)\geq n\ .
$$

For the converse direction, consider a word $u=b_1\dots b_k$ over $\alphabetB$ such that 
$$\sem{\monoid_\coideal,H,K}_{3|M_\coideal|}^{--}(u)\leq n\ .$$
This means that there exists an $n$-under-computation for $A_1\dots A_k=\tilde H(b_1\dots b_k)$ of value $J\not\in K$,
and height at most $3|M_\coideal|$. 
Let $v=c_1\dots c_k$ be some word over $\alphabet$ such that $\tilde z(v)=u$.
By definition of $H$, this means that $a_i=h(c_i)\in H(b_i)$ for all $i=1\dots k$.
Thus, by Lemma~\ref{lemma:sup-key2}, there exists an $n$-computation for $a_1\dots a_k$
of value in $a\in J$ of height at most $d=9|M||M_\coideal|$.
Since $J\not\in K$, this means that $J\cap I=\emptyset$. As a consequence $a\in M\setminus I$.
It follows that $\sem{\monoid,h,I}^{-}_d(v)\leq n$. Since this holds for all $v$ such that $\tilde z(v)=u$, we get
$$(\sem{\monoid,h,I}^{-}_d)_{\sup,z}(u)\leq n\ .$$
\end{proof}

\section{On the role of automata}
\label{section:conclusion}

In this paper we have developed the algebraic and logical aspects of regular cost functions over finite words. 
More precisely, we have introduced a notion of logic, cost monadic logic, that is suitable for describing functions,
and an algebraic notion of stabilisation monoid that is suitable for recognising functions up to an equivalence relation $\approx$.
We have shown that the logically defined functions could be translated into equivalent ones recognisable by
stabilisation monoids. Decision procedures for several problems involving the 
existence of upper bounds for functions are derived from this translation.

There could have been several other facets for approaching this theory, a very natural one being through automata.
The automata theoretic presentation happens to be closer to the historical developments.
Indeed, the study of distance automata \cite{Hashiguchi82a}, and then
of nested distance desert automata \cite{Kirsten05} was the original motivation.
Following ideas from \cite{LICS06:bojanczyk-colcombet}, it is convenient to consider two dual forms of automata
using counters, called $\mathtt B$ and $\mathtt S$-automata. The first model computes a minimum over all runs
of the maximal values taken by counters, and the second form computes a maximum over all runs of the
minimum value taken by some counters at some identified places in the run.
As for regular languages, these automata happen to have the same expressiveness for describing cost functions
as stabilisation monoids.

Technically, all the necessary material for proving the equivalence between automata and regular cost functions
is already present in this paper. Indeed, in one direction, as it is classical for regular languages, automata can be seen
as a special fragment of cost monadic logic. Thus, automata can only define regular cost functions. 
For the converse implication, it is easy to construct a $\mathtt B$-automaton guessing under-computations,
or an $\mathtt S$-automaton guessing over-computations, and use it for describing a recognisable 
cost function. We have seen all the necessary material for establishing the correction of these approaches.

Despite this strong connection, there are several reasons for not presenting automata in this document.

A first reason is to emphasize the difference with the theory of regular languages. In the case of languages,
the simplest way to show the decidability of monadic logic over words is to use automata. This is not the case anymore here.
Proving the important results concerning $\mathtt B$ and $\mathtt S$-automata (the central one being the
equivalence between the two models, called the duality theorem), is more complicated than developing the
theory of stabilisation monoids. In fact, the simplest way to prove the duality theorem
is to translate the (say) $\mathtt B$-automaton into a stabilisation monoid, and only then into an $\mathtt S$-automaton
(though, some other techniques are possible). One explanation for this difference between the theory of
regular languages and regular cost functions is that $\mathtt B$-automata and $\mathtt S$-automata cannot be determinised.
For these reasons stabilisation monoids form a much simpler model.

A second reason is that we could concentrate even more deeply on the model of stabilisation monoid. In particular,
we did not only develop stabilisation monoids for obtaining decision procedures (as all works using stabilisations were doing so far),
but we proved that a suitably axiomatised notion of stabilisation monoid can be used to recognise a cost function
independently of the presence of any cost monadic formula, or any automaton. This is reminiscent of the proof 
in the theory or regular languages of infinite words that
finite Wilke algebras can be translated in a unique way into $\omega$-semigroups.
If we were only interested in decidability
questions, the paper could be simplified, and the important Theorem~\ref{theorem:unicity-computations} omitted.

A third reason is that $\mathtt B$-automata and $\mathtt S$-automata, which may seem a bit useless under
the light of the previous explanations, are in fact so important that they require a deep study on their own.
The importance of automata does not stem from the question of decidability of cost monadic logic over words, but over trees (even finite)\cite{LICS10:colcombet-loeding}.
Indeed, the situation is reminiscent from the case of regular languages of infinite trees. In this case,
proving the decidability of monadic logic over infinite trees can only be achieved using infinite tree automata,
and the proof makes also use of the ability to determinise automata over infinite words (see, e.g, the survey \cite{Thomas97}).
The situation is similar here, and what is important in the study of automata is to disclose a suitable variant of the
notion of determinism, called history-determinism \cite{ICALP09:colcombet-cost-function}, and to prove that $\mathtt B$
and $\mathtt S$-automata can be made history-deterministic. These considerations are completely diverging from the
content of this paper.

\subsection*{Acknowkedgement}

This paper has gained a lot from the 
{discussions}	 with and reviewing from
Achim Blumensath, Michael Vanden Boom, Denis Kuperberg and Christof L\"oding.
I am also very grateful to the two anonymous reviewers for their constructive remarks and
their very thorough reading of the document. I also thank Wolfgang Thomas for his patient and careful editing work.

\bibliographystyle{plain}
\bibliography{../../../../../main-biblio}
\end{document}